\documentclass[11pt,a4paper]{article}
\usepackage[margin=1in]{geometry}
\usepackage{tikz}
\usetikzlibrary{positioning,arrows.meta}
\usepackage{subcaption}
\usepackage{makecell}
\usepackage[english]{babel}
\usepackage{array} % 提供 p{} 功能
\usepackage{booktabs} % 可选，提供更美观的表格线

\usepackage{xcolor}

\usepackage{authblk}
\usepackage[colorlinks=true, allcolors=blue]{hyperref}

\usepackage{multirow}

\usepackage[ruled,linesnumbered,vlined]{algorithm2e}
\usepackage{algpseudocode}
\usepackage{amssymb}
\usepackage{amsmath}
\usepackage{amsthm}
\usepackage{graphicx}
\usepackage[colorlinks=true, allcolors=blue]{hyperref}
\usepackage{natbib}
\usepackage{framed}
\definecolor{shadecolor}{gray}{0.9}
\usepackage{pgfplots}
\pgfplotsset{compat=1.17}
\usepackage{tcolorbox}
\usepackage{tikz}
 \usetikzlibrary{patterns}
\usepackage{pgfplotstable}
    \pgfplotsset{
    name nodes near coords/.style={
        every node near coord/.append style={
            name=#1-\coordindex,
            alias=#1-last,
        },
    },
    name nodes near coords/.default=coordnode
    }

\newtheorem{theorem}{Theorem}[section]
\newtheorem{definition}[theorem]{Definition}
\newtheorem{claim}[theorem]{Claim}
\newtheorem{corollary}[theorem]{Corollary}
\newtheorem{lemma}[theorem]{Lemma}
\newtheorem{invariant}[theorem]{Invariant}

\newcommand{\E}{\mathbb{E}}
\newcommand{\bP}{\mathbf{P}}
\newcommand{\bX}{\mathbf{X}}

\newcommand{\MMS}{\mathsf{MMS}}

\title{MMS Allocation for Chores with Online Agent Arrivals}

\author[1]{Haolong Li}
\author[1]{Zehan Lin}
\author[1,2]{Huahua Miao}
\author[1]{Xiaowei Wu}

\affil[1]{University of Macau\\
\texttt{\{yc47435,yc47490,yc47994,xiaoweiwu\}@um.edu.mo}}

\affil[2]{Shenzhen Institutes of Advanced Technology,
Chinese Academy of Sciences}

\date{}

\begin{document}
\maketitle

\begin{abstract}
We study the fair allocation of $m$ indivisible chores to $n$ agents with subadditive cost functions arriving online in an arbitrary order. Upon an agent’s arrival, we are informed of her cost function and must irrevocably assign her a set of chores. We focus on the Maximin Share (MMS) fairness notion and aim to compute an allocation in which all items are assigned, and no agent incurs a cost more than $\alpha$ times her MMS.

Without any prior information about the instance (other than $n$ and $m$), we design an algorithm with a competitive ratio of $O(\min\{n, k\log^{1+\epsilon}k, \log m\})$ for any constant $\epsilon > 0$, where $k$ denotes the number of cost function types. 
Our bound matches the best known offline approximation guarantees for MMS under subadditive costs and is nearly optimal with respect to all three parameters: we show that even for binary additive cost functions, no online algorithm can achieve a competitive ratio of $o(\min\{n, k\log k, \log m\})$.

We then consider the setting in which the $k$ cost function types are known in advance (though the realized types of arriving agents are not). 
For additive cost functions, we provide an algorithm with a competitive ratio of $O(\min\{\log k, \log(kn)/\log\log(kn)\})$, and show that constant-competitive algorithms do not exist for general $k$, even for the binary additive setting. For binary additive functions when $k \le n$, we propose a $3$-competitive algorithm and establish a lower bound of $2$.
\end{abstract}

\newpage

\section{Introduction}
\label{sec:intro}

Fair allocation studies how to allocate a set $M$ of $m$ indivisible items among a group $N$ of $n$ agents with heterogeneous preferences. 
Depending on these preferences, items may be viewed either as desirable goods (e.g., resources) or undesirable chores (e.g., tasks). 
The primary objective in this area is to compute allocations that satisfy well-defined notions of fairness.
In this work, we focus on the fair allocation of chores, where each agent evaluates bundles via a non-negative cost function. 
Our main fairness concept is the \emph{maximin share} (MMS). 
Originally introduced by Budish~\cite{conf/bqgt/Budish10} for goods, this notion extends naturally to chores, where the MMS value of an agent is the minimum cost she can guarantee herself by partitioning all chores into $n$ bundles and then receiving the worst bundle (i.e., the bundle with the highest cost). 
An allocation is said to be $\alpha$-MMS, for some approximation ratio $\alpha \ge 1$, if every agent receives a bundle whose cost is at most $\alpha$ times her MMS value.

Prior work on chore allocation has mostly focused on the offline setting, where the entire set of agents and items is known in advance. 
However, many real-world applications are inherently dynamic, with information revealed sequentially.
This motivates the study of fair allocation in online settings. 
For example, Zhou et al.~\cite{icml/ZhouBW23} studied MMS approximations in a model where agents are offline, but items arrive sequentially and must be allocated irrevocably upon arrival. 
They showed that no online algorithm can achieve a competitive ratio better than $2$ without prior information. 
Song et al.~\cite{corr/abs-2507-14039} later strengthened this lower bound to $\Omega(n)$ and provided an $O(\min\{n, k\})$-approximation algorithm in the $k$-valued setting, where the cost function of each agent has at most $k$ different values.

Motivated by disaster relief scenarios, where supplies (like food or medical resources) must be distributed promptly as requests arrive, Kulkarni et al.~\cite{sigecom/KulkarniMS25} studied a complementary model in which items are given offline, and agents arrive sequentially. 
They focused on the allocation of goods and proved that no positive-MMS guarantee is achievable without prior information, even for two agents. 
To overcome this limitation, they considered the (known) $k$-type setting, where the algorithm is given a set of $k$ valuation functions, and each arriving agent reveals her type upon arrival. 
They designed a $(1/k)$-MMS competitive algorithm and complemented it with a lower bound of $\Omega(1/\sqrt{k})$.

In this work, we investigate the analogous online setting for chores, which naturally arises in dynamic task assignment problems. 
Consider the following illustrative example.

\begin{tcolorbox}[colback=gray!15, colframe=gray!15]
\textbf{Task Assignment to Volunteers.} 
Consider a campus volunteer event in which an organizer must recruit a fixed number of volunteers to complete a predefined set of tasks. 
Volunteers arrive sequentially; upon arrival, the organizer learns each volunteer's skills and preferences and must irrevocably assign a subset of tasks. 
The organizer's goal is to complete all tasks while ensuring that no volunteer feels unfairly treated.
\end{tcolorbox}

In this setting, the organizer knows the total number of volunteers in advance but learns each agent’s cost function only upon arrival, and the assignments of tasks are irrevocable. 
This motivates the online model that we formally define next.
\paragraph{Chore Allocation to Online Agents.}
In this model, a set of $m$ items is given offline, while $n$ agents arrive online sequentially; the value of $n$ is known to the algorithm. 
Upon the arrival of an agent $i \in N$, the algorithm learns her cost function $c_i$ and must irrevocably assign a subset of unallocated items to her. 
The goal is to compute a full allocation (in which all items are allocated) that minimizes the MMS approximation ratio.

\medskip

In this work, we assume that agents have subadditive cost functions and focus only on deterministic online algorithms. 
In contrast to the case of goods, for chore allocation, any online algorithm trivially achieves a competitive ratio of at most $n$, since every allocation is $n$-MMS. 
This raises the natural question of whether strictly better competitive ratios can be obtained. 
Unfortunately, as we show next, this is not possible in general.

\paragraph{Hardness of the Problem.}
Consider, for example, the case of two agents. 
Suppose the first agent has a unit cost for each of the $m$ chores. 
Unless the algorithm assigns more than $m-2$ chores to the first agent, there exists an instance in which the second agent has cost $1$ on two of the remaining chores and cost $0$ on the other items, for which the allocation is not better than $2$-MMS. 
On the other hand, if more than $m-2$ chores are assigned to the first agent, then her MMS approximation ratio approaches $2$ as $m \to \infty$.
This hard instance naturally extends to $n$ agents. 
In particular, consider binary cost functions with a nested structure: each arriving agent has cost $1$ on all currently unallocated chores and cost $0$ on all previously allocated ones. 
For sufficiently large $m$, this construction implies that no online algorithm can achieve a competitive ratio strictly better than $n$.

\medskip

The above hardness result shows that beating the competitive ratio of $n$ is impossible in general. 
However, the construction relies on instances with a very large number of chores (e.g., exponential in $n$) and entirely heterogeneous agent cost functions, which may be unrealistic in practice. 
For example, in the above task assignment example, the number of tasks is typically bounded, and the volunteers naturally fall into a small number of types based on their skills or preferences. 
These observations motivate the following question.

\begin{quote}
    \emph{Can we obtain competitive ratios strictly better than $n$ when the number of chores is bounded, or when agents belong to a small number of types?}
\end{quote}

Furthermore, following prior work~\cite{sigecom/KulkarniMS25} on online fair allocation with type information, we investigate whether additional knowledge about cost functions can lead to improved guarantees. 
In particular, we ask:

\begin{quote}
    \emph{Can we design algorithms with improved (e.g., polylogarithmic or constant) competitive ratios when the types of cost functions are known?}
\end{quote}

In this work, we answer both questions in the affirmative.

\subsection{Our Contributions}
We begin with the fully uninformed setting, in which the algorithm knows only $n$ and $m$ and has no prior knowledge about the agents' cost functions.
We show that, despite this lack of information, there exists an algorithm with a competitive ratio strictly better than $n$ when either the number of items is bounded or the agents belong to a small number of types.
Moreover, the achieved competitive ratio is nearly optimal.

\begin{tcolorbox}[colback=gray!15, colframe=gray!15]
{\bf Result 1} (Theorem \ref{thm:fully-uninformed-upper}){\bf .} \label{Result1}
{\em For the allocation of chores to online agents with subadditive cost functions, if the agents belong to $k$ unknown types, there exists an algorithm that guarantees an $O\!\left( \min\!\left\{ n, k\log^{1+\epsilon}k,\log m \right\} \right)$ competitive ratio for any constant $\epsilon > 0$.}
\end{tcolorbox}

Our algorithm is near-optimal in two respects.
First, its competitive ratio matches the best known offline approximation ratio for MMS under subadditive cost functions due to Li et al.~\cite{conf/nips/0037WZ23}, who also proved a nearly matching lower bound of $\Omega(\min\{n,\log m/\log\log m\})$.
We further adapt their analysis to establish a lower bound of $\Omega(k)$ when the agents belong to $k$ types.
Second, we prove a nearly matching lower bound of $\Omega\left(\min\{n, k\log k, \log m\}\right)$ for all deterministic online algorithms, even for binary additive instances.

Our algorithm is based on the simple idea of asking each online agent to take sufficiently many items so that all items are eventually allocated.
Moreover, by using the same MMS partition for agents of the same type and assigning different bundles to different agents, we can guarantee that all items are allocated when sufficiently many agents of that type arrive.
The main challenge is to achieve an $O(k\log k)$ competitive ratio when the number $k$ of agent types is unknown.
For example, when the first agent arrives, we can allocate to it only $O(1)$ bundles from her MMS partition; otherwise, in the case where all remaining agents have the same type as agent~1 (i.e., $k=1$), the resulting approximation ratio would be $\omega(k\log k)$.
To overcome this difficulty, we maintain an estimate of the number of agent types and use it to dynamically adjust the number $\ell_i$ of bundles that each online agent $i$ takes from her MMS partition.
With a carefully designed choice of $\ell_i$, we show that all items are allocated while the MMS approximation ratio remains bounded.

\medskip

We next turn to the known-types setting studied in previous work, where the algorithm is given a set of $k$ cost functions (which are the types) but learns the type of an online agent only upon her arrival.
We investigate whether strictly better competitive ratios can be achieved in this setting.
Unfortunately, as discussed above, the hardness of approximation in the offline setting~\cite{conf/nips/0037WZ23} leaves very little room for improvement under subadditive cost functions.
We therefore focus on additive and binary additive cost functions.
We emphasize that, in the known-types setting, not all agent types are guaranteed to appear.
This distinction is particularly important in the binary setting, since otherwise the problem becomes trivial: we could simply wait and assign each item $e$ to an arriving agent $i$ with $c_i(e)=0$, thereby reducing the problem to allocating only those items for which every type incurs cost~$1$, which is a trivial problem.
Consequently, our model allows $k > n$.
Our next two results show that the MMS guarantee can be improved substantially in the known-types setting.

\begin{tcolorbox}[colback=gray!15, colframe=gray!15]
{\bf Result 2} (Theorem \ref{thm:known-types-additive-upper}){\bf .} \label{Result2}
{\em For the allocation of chores to online agents with additive cost functions, if the agents belong to $k$ known types, there exists an algorithm that guarantees an $O\!\left(\min\left\{ \frac{\log(kn)}{\log\log(kn)}, \log k \right\}\right)$ competitive ratio.}
\end{tcolorbox}
This result demonstrates that logarithmic competitive ratios are achievable in the known-types setting and establishes a sharp separation between additive and subadditive cost functions.

Our main algorithmic idea is to reduce the online allocation problem to computing an offline partition of items that is fair for all types. Specifically, we seek a partition of the items into $n$ bundles such that the partition is $\alpha$-MMS with respect to every cost function type. We refer to such a partition as an $\alpha$-approximate \emph{universal partition}. Clearly, if an $\alpha$-approximate universal partition can be computed, then assigning each arriving agent an arbitrary unallocated bundle from the partition yields an $\alpha$-MMS allocation.
Complementing this positive result, we provide a hard instance showing that the best achievable approximation factor for universal partitions with $k$ cost functions is $\Omega\!\left(\frac{\log k}{\log\log k}\right)$.

\smallskip

Although binary additive costs have a simple structure, all of our hardness results already hold for binary instances. We therefore study this setting to identify the core difficulty behind the general additive problem.
Unfortunately, a direct consequence of the $\Omega(\log m)$ hardness result for the uninformed setting implies that no online algorithm can guarantee an $o(\log \log k)$ competitive ratio in the known-types binary additive setting.
In particular, constant-competitive algorithms do not exist for general $k$.
We therefore focus on the regime $k \leq n$ and develop a $3$-competitive algorithm.

\begin{tcolorbox}[colback=gray!15, colframe=gray!15]
{\bf Result 3} (Theorem \ref{thm:known-binary-three-mms}){\bf .} \label{Result3}
{\em There exists a $3$-competitive algorithm for the allocation of chores to online agents with binary additive cost functions, provided that the agents belong to $k$ known types and $k \leq n$.}
\end{tcolorbox}

Unlike a universal partition, our binary algorithm chooses each bundle after the current type is revealed. It controls the cost for the receiving agent while keeping the remaining chores safe for every type. We call this a \emph{universal residual}.
We first show that a certain randomized allocation process achieves a bounded MMS approximation ratio with positive probability, and then derandomize the process via an appropriately designed potential function.
However, there is a crucial difference: rather than requiring, as universal partitions do, that every bundle has small cost under every cost function type, we relax this condition.
Specifically, we compute the bundles on-the-fly as agents arrive, and we require only that the allocated bundle has a small cost for the receiving agent, while the remaining set of items maintains a small cost for all types.
We refer to this property as maintaining a \emph{universal residual}.
In the binary additive setting, this can be achieved by assigning a random bundle of fixed size determined by the arriving agent's cost function.
We prove that such random allocations preserve the universal residual property with positive probability. Consequently, by defining a suitable potential function over the residual set of items, we obtain a polynomial-time derandomization of the process.
Complementing this positive result, we show that exact MMS allocations remain unattainable even in this restricted regime: no deterministic algorithm can guarantee better than a $2$-MMS allocation even when $k=O(1)$.

\subsection{Technical Overview}

We provide an overview of the main techniques used in the known-types setting.

\paragraph{Universal partitions.}
Although all possible cost-function types are known in advance, the type of each arriving agent is revealed only upon arrival. Thus, an allocation decision made for the current agent may leave an unfavorable set of chores for future agents. Our first approach avoids this online uncertainty by constructing, before any agent arrives, a partition that is approximately MMS for every possible type. We call such a partition a \emph{universal partition}. Once it is computed, each arriving agent can simply receive an arbitrary unused bundle, regardless of her type or the arrival order.
We give two constructions. The first assigns every item independently and uniformly to one of the $n$ bundles. Since the desired cost bound must hold for every type and every bundle, there are $kn$ constraints to control. A potential-function analysis gives an $O\left(\frac{\log(kn)}{\log\log(kn)}\right)$ guarantee, and the random process can be derandomized by conditional expectation.
To remove the dependence on $n$, we instead construct the bundles one at a time. In each round, we select a bundle whose cost is small for every type while ensuring that the remaining chores still have sufficiently small total cost for every type. Thus, each round only needs to control the $k$ types rather than all $kn$ type-bundle pairs. A two-sided potential function guarantees that such a suitable bundle exists and can be found deterministically, which gives an $O(\log k)$ approximation.

\paragraph{Universal residuals.}
For binary additive costs, we can further exploit the type revealed by the current agent. Unlike a universal partition, we no longer require the current bundle to have small cost for every possible type. It only needs to have small cost for the receiving agent. At the same time, we keep the remaining chores sufficiently small under every type, so that future agents can still be served fairly. We call this property a \emph{universal residual}.
Concretely, when an agent arrives, we first allocate all remaining chores that have zero cost for her. From the remaining costly chores, we select a fixed-size subset for her. The size is chosen so that her cost is at most three times her MMS. We define a potential function that measures whether the residual set may become too large for any type. A uniformly random choice of the fixed-size subset does not increase this potential in expectation. We then derandomize this choice by conditional expectation, which ensures that the universal residual is maintained in every round and gives a deterministic $3$-competitive algorithm when $k\leq n$.

\subsection{Related Work}

Given the vast literature on fair allocation, we focus on the results most closely related to our work. For a more comprehensive overview, we refer readers to the surveys by Aleksandrov and Walsh~\cite{conf/aaai/AleksandrovW20} and Amanatidis et al.~\cite{journals/ai/AmanatidisABFLMVW23}.

\paragraph{Online Fair Allocation for Chores.}
Most prior work in this domain focuses on the item-arrival model.
Without prior information, Zhou et al.~\cite{icml/ZhouBW23} showed that no deterministic algorithm can do better than $2$-MMS, though knowing the total cost in advance enables a $(2-1/n)$-MMS guarantee.
Subsequent works further established $\Omega(n)$ impossibility results in the absence of prior information~\cite{corr/abs-2507-14039,journals/corr/abs-2507-12984}.
Neoh and Teh~\cite{neoh2026closing} further extended these impossibility results to randomized algorithms against adaptive adversaries, showing that no $\alpha$-MMS guarantee is possible for any $1\leq\alpha<n$.
Song et al.~\cite{corr/abs-2507-14039} designed an online algorithm that achieves a $\min\{n,O(k),O(\log D)\}$-MMS approximation, where the cost functions of agents are $k$-valued, and $D$ is the ratio between the maximum cost and the minimum cost.
Wang and Wei~\cite{journals/corr/abs-2505-24321} studied binary additive costs and showed that exact MMS allocations can be guaranteed online.
For the more general personalized bivalued setting, where costs are in $\{1,k_i\}$ with $k_i>1$, they established a $5/3$-MMS guarantee for two agents.
This positive result was extended to an arbitrary number of agents with a $(2+\sqrt{3})$-MMS approximation by Song et al.~\cite{corr/abs-2507-14039}.
%\rev{For proportionality up to one chore, Neoh and Teh~\cite{neoh2026closing} showed that no online algorithm can achieve a factor below $n$ against adaptive adversaries.}\rev{For mixed manna, Choi et al.~\cite{choi2026temporal} gave an online algorithm that maintains EF1 after every item arrival when there are only two types of items.} \rev{Amanatidis et al.~\cite{amanatidis2026buffers} obtained EF1 guarantees for personalized finite-value mixed-manna instances by allowing a bounded buffer to delay and reorder item allocations.}

\paragraph{Online Fair Allocation for Goods.}
The goods setting has been studied under both item- and agent-arrival models. 
Closest to our work is that of Kulkarni et al.~\cite{sigecom/KulkarniMS25}, who studied MMS allocations for online arriving agents. 
They studied the (known) $k$-type setting and established a $1/k$ lower bound and a $2/(\sqrt{k}-2)$ upper bound for MMS approximations under the adversarial arrival setting. 
En et al.~\cite{journals/corr/abs-2606-18679} also investigated online resource allocation under the agent-arrival model.
Instead of allocating items to agents, their model allocates online agents to offline facilities to maximize the total utility. 
For the item arrival model, existing work also investigated vanishing envy over time~\cite{conf/sigecom/BenadeKPP18}, and fairness and efficiency trade-offs under random distributions~\cite{conf/sigecom/ZengP20}. Additionally, a growing line of literature focuses on online allocation with predictions or prior knowledge regarding future items~\cite{banerjee2022online,journals/corr/abs-2505-24503,chen2026competitiveanalysisonlinefair}.
From a machine learning perspective, Procaccia et al.~\cite{conf/nips/ProcacciaS024} and Yamada et al.~\cite{conf/aistats/YamadaKAI24} examined online fair allocation with indivisible items through bandit learning frameworks. 
Moreover, a separate line of literature on divisible resources focused on maintaining fairness during sequential agent arrivals with minimal allocation updates~\cite{jair/KashPS14,ec/FriedmanPV17,moor/VardiPF22,conf/ijcai/LiLL18}.

\paragraph{Offline Approximate MMS for Chores.}
For offline chore allocation, exact MMS allocations exist for two agents but need not exist for $n\ge 3$~\cite{conf/wine/FeigeST21,conf/aaai/AzizRSW17}. 
For additive cost functions, the MMS approximations have been extensively studied under both cardinal and ordinal preferences~\cite{conf/aaai/AzizRSW17,journals/mp/AzizLW24,conf/sigecom/FeigeH23,journals/teco/BarmanK20,conf/sigecom/HuangL21}.
The best known cardinal approximation ratio is $13/11$~\cite{conf/sigecom/HuangS23}, and a lower bound of $44/43$ is known even for three agents~\cite{conf/wine/FeigeST21}.
In ordinal models, the known guarantees are weaker, improving from $2-1/n$~\cite{conf/aaai/AzizRSW17} to $5/3$~\cite{journals/mp/AzizLW24} and then to $8/5$~\cite{conf/sigecom/FeigeH23}. 
For submodular chore allocation, Li et al.~\cite{conf/nips/0037WZ23} established an $\Omega(\min\{n,\log m/\log\log m\})$ lower bound and designed an algorithm that computes an $O(\min\{n,\log m\})$-MMS allocation for subadditive cost functions.

\section{Preliminaries}
\label{sec:preliminary}

For every positive integer $q$, let $[q]=\{1,\ldots,q\}$. We study the allocation of a set $M$ of $m$ indivisible chores to a set $N=[n]$ of agents who arrive online. When agent $i$ arrives, her cost function $c_i:2^M\to\mathbb R_{\geq0}$ is revealed, and the algorithm must irrevocably assign her a bundle $X_i$ of unallocated chores. The resulting allocation $\mathbf X=(X_1,\ldots,X_n)$ is an ordered $n$-partition of $M$.
For convenience we write $c(e):=c(\{e\})$ for cost function $c$ and item $e\in M$. Unless stated otherwise, we assume that all cost functions satisfy $c(\varnothing)=0$ and are monotone.
For agents sharing the same cost function, we say that they belong to the same \emph{type}.

In this paper, we study different classes of cost functions.

\begin{definition}[Subadditive]\label{def:subadditive}
A cost function $c:2^M\to \mathbb{R}_{\ge 0}$ is \emph{subadditive} if
\begin{equation*}
    c(S\cup T)\le c(S)+c(T), \quad \forall S,T\subseteq M.
\end{equation*}
\end{definition}

\begin{definition}[Additive]\label{def:additive}
A cost function $c:2^M\to \mathbb{R}_{\ge 0}$ is \emph{additive} if
\begin{equation*}
    c(S) = \sum_{e \in S} c(e), \quad \forall S\subseteq M.
\end{equation*}
\end{definition}

\begin{definition}[Binary Additive]
\label{def:binary}
An additive cost function $c:2^M\to \mathbb{R}_{\ge 0}$ is \emph{binary} if for every item $e \in M$, we have $c(e) \in \{0, 1\}$.
\end{definition}

In this work, we focus on the fairness notion of (approximate) maximin share, which we formally define as follows.

\begin{definition}[Maximin Share]
    Given a set of chores $M$ and $n$ agents, the maximin share (MMS) value of agent $i$ is defined as
    \begin{equation*}
        \MMS_i (M) :=\min_{\bX \in \prod_n(M)} \max_{j \in N}  \ \left\{c_i(X_j)\right\},
    \end{equation*}
    where $\prod_n(M)$ denotes the set of all $n$-partitions of $M$.
\end{definition}

We use
$\MMS_i$ to denote $\MMS_i(M)$ when $M$ is clear from the context.
We call an $n$-partition $\mathbf{P} = (P_1,P_2,\ldots,P_n)$ of $M$ an \emph{MMS partition} of agent $i$ if $c_i(P_j)\leq \MMS_i$ holds for all $j\in [n]$.
Note that there might be multiple MMS partitions for a given cost function.
However, we can fix one of them arbitrarily, and therefore, in this paper, we assume that there exists a \emph{unique} MMS partition for every given cost function.
This is crucial when multiple agents share the same cost function.

\begin{definition}[$\alpha$-MMS]
    An allocation $\bX$ is said to satisfy $\alpha$-approximate maximin share guarantee for some $\alpha \geq 1$ if for every agent $i \in N$, we have $c_i(X_i) \leq \alpha \cdot \MMS_i$.
    In particular, when $\alpha = 1$, $\bX$ is an \emph{MMS} allocation.
\end{definition}

By convention of online algorithm analysis, we say that an algorithm is $\alpha$-competitive if it computes an $\alpha$-MMS allocation for any online instance.
In the remainder of this paper, we use $\log$ to denote the natural logarithm (with base $e$).

\section{Unknown Subadditive Cost Functions}
\label{sec:fully-uninformed}

In this section, we study the fully uninformed setting, where the online algorithm knows only $n$ and $m$. We assume that agents have subadditive cost functions. Let $k$ denote the number of types among the online agents. Note that the online algorithm does not know $k$ in advance. We present an algorithm that guarantees an $O(\min\{n,\, k\log^{1+\epsilon}k,\, \log m\})$-MMS allocation for any constant $\epsilon>0$. 
Moreover, we establish nearly matching lower bounds on the MMS guarantee achievable, showing that our algorithm is nearly optimal.

\subsection{The Online Algorithm}
\label{sec:fully-uninformed-upper}

We first present the algorithm and prove an upper bound of $O(\min\{n,\, k\log^{1+\epsilon}k,\, \log m\})$ on its competitive ratio. Since every allocation is trivially $n$-MMS, we focus on the bounds that depend on $k$ and $m$.
At a high level, the algorithm asks each online agent to take a sufficiently large number of items so that only a limited number remain when the last agent arrives. Furthermore, we exploit the $k$-type structure by assigning different bundles in the MMS partition (recall that we fix a unique MMS partition\footnote{Since the computation of an MMS partition is NP-hard in general, if computational complexity matters, we can compute an $O(1)$-approximation of the MMS partition in polynomial time. It can be verified that this does not affect the asymptotic competitive ratio of our algorithms.} for each cost function) to agents of the same type. Consequently, when sufficiently many agents of a given type appear, we can ensure that all items are eventually allocated. For example, when all agents belong to a single type (i.e., $k=1$), we can compute an MMS partition and assign a distinct bundle to each agent, thereby obtaining an exact MMS allocation. However, since the value of $k$ is unknown in advance, the algorithm must proceed adaptively: it updates its estimate on $k$ based on the types observed so far and requires agents to take more bundles as additional types are discovered.

Motivated by these ideas, our algorithm (Algorithm~\ref{alg:main-alg}) operates as follows.

For each agent $i$, let $k_i$ denote the number of distinct cost functions observed up to and including agent $i$. The algorithm then asks agent $i$ to take $\ell_i$ bundles from the MMS partition of $c_i$, where the selected bundles are those containing the largest numbers of currently unallocated items. The value of $\ell_i$ is determined as a function of $k_i$.
Specifically, we fix a small constant $\epsilon > 0$, define $\delta_\epsilon:=\frac{3}{2\log^{1+\epsilon} 2} + \frac{1}{\epsilon\log^\epsilon2}$, and define function
\begin{equation*}
    f(r):=2 \lceil\delta_\epsilon\cdot r\cdot \log^{1+\epsilon}(r+1)\rceil.
\end{equation*}
Let $L:=2(\lceil\log m\rceil+1)$ and $\ell_i=\min\{n,f(k_i),L\}$.

\begin{algorithm}[!htb]
\caption{$O(\min\{n, k\log^{1+\epsilon}k,
\log m\})$-MMS allocations for subadditive costs}
\label{alg:main-alg}
\KwIn{A set $M$ of items, the number of online agents $n$, and a fixed constant $\epsilon>0$}

Let $R_1 \gets M$ be the set of remaining items.

\For{each online agent $i$, where $i=1,2,\ldots,n$}{
    Compute the MMS partition $\mathbf P^{(i)}=(P^{(i)}_1,\ldots,P^{(i)}_n)$ with respect to $c_i$\;

    Let $k_i$ be the number of arrived types so far\;
    
    $\ell_i\gets\min\{n, f(k_i),L\}$;

    Relabel the bundles of $\mathbf P^{(i)}$ so that
    $|P^{(i)}_1\cap R_i|\geq\cdots\geq|P^{(i)}_n\cap R_i|$;

    Let $X_i\gets
    \bigcup_{j=1}^{\ell_i}(P^{(i)}_j\cap R_i)$ be union of the $\ell_i$ bundles with most unallocated items;

    Update the set of remaining items: $R_{i+1} \gets R_i\setminus X_i$;
}
\KwOut{Allocation $\mathbf X=(X_1,\ldots,X_n)$}
\end{algorithm}

We prove the following main result of this section.

\begin{theorem}
\label{thm:fully-uninformed-upper}
For any constant $\epsilon>0$, Algorithm~\ref{alg:main-alg} returns an $O( \min\!\left\{n,k\log^{1+\epsilon}k,\log m\right\})$-MMS allocation for subadditive cost functions.
\end{theorem}

It is straightforward to establish the approximation ratio of the allocation by the monotonicity and subadditivity of the cost functions.
Consider an arbitrary agent $i$. We have
\begin{align*}
    c_i(X_i) &\leq \sum_{j=1}^{\ell_i}c_i(P^{(i)}_j\cap R_{i}) \leq \sum_{j=1}^{\ell_i}c_i(P^{(i)}_j) \leq \ell_i\cdot \MMS_i.
\end{align*}

Therefore, the approximation ratio of MMS is $\max_{i\in N}\{\ell_i\} = O( \min\!\left\{n,k\log^{1+\epsilon}k,\log m\right\})$, where $k = k_n$ is the total number of agent types.
It remains to show that all items are allocated.

We first establish a technical lemma, which (intuitively speaking) states that the function values of $f$ are sufficiently large to ensure a full allocation.

\begin{lemma}
\label{lemma:function-f}
The function $f$ is nondecreasing and satisfies
$$
\sum_{r=1}^{\infty}\frac{1}{f(r)}\leq\frac{1}{2}.
$$
\end{lemma}

\begin{proof}
The function $f$ is nondecreasing by definition. 
Furthermore, we have 
\begin{equation*}
    \sum_{r=1}^{\infty}\frac{1}{f(r)} \leq \frac{1}{2\delta_\epsilon} \cdot \sum_{r=1}^{\infty} \frac{1}{r\cdot \log^{1+\epsilon}(r+1)}.
\end{equation*}

For $r\geq2$, we have $\log(r+1)\geq\log r$. The integral test gives
\begin{align*}
    \sum_{r=1}^{\infty} \frac{1}{r\log^{1+\epsilon}(r+1)}
    &\leq \frac{1}{\log^{1+\epsilon}2} + \frac{1}{2\log^{1+\epsilon}2} + \sum_{r=3}^{\infty} \frac{1}{r\log^{1+\epsilon} r}\\
    &\leq \frac{3}{2\log^{1+\epsilon}2} + \int_2^\infty \frac{\mathrm{d}x}{x\log^{1+\epsilon} x}
    = \frac{3}{2\log^{1+\epsilon}2} + \frac{1}{\epsilon\log^\epsilon2} = \delta_\epsilon.
\end{align*}

Combining the two inequalities yields the lemma.
\end{proof}

Using the property of $f$, we show that all items are allocated.

\begin{lemma} \label{lemma:all_item_allocated_uninformed}
    All items are allocated when Algorithm~\ref{alg:main-alg} terminates.
\end{lemma}
\begin{proof}
Suppose, for contradiction, that some items are left unallocated after all agents have arrived. 
This implies that every agent $i$ takes exactly $\ell_i$ non-empty bundles $P^{(i)}_j\cap R_i$ when she arrives.
Recall that $\ell_i = \min\{n, f(k_i),L\}$ and $k_i$ is non-decreasing in $i$.
Therefore, $\ell_i$ is non-decreasing and upper bounded by $L$.
Furthermore, since some items are unallocated after all agents have arrived, we must have $\ell_i \neq n$ for all $i$. 
We let $i^*$ be the first agent for which $\ell_i = L$.
If no such agent exists, we let $i^* = n+1$.
By definition, we have $\ell_i = f(k_i)$ for all agents $i < i^*$; and $\ell_i = L$ for all agents $i \geq i^*$.
We first show that $i^* \leq n/2$.

\smallskip

Fix an arbitrary type $t$ and let $N_t$ be the set of agents of type $t$ that arrive before $i^*$, and let $i_t\in N_t$ be the agent who arrives first.
Note that all agents in $N_t$ share the same cost function, and we use the same MMS partition whenever any of these agents arrives.
By the design of the algorithm, the total number of bundles taken by agents in $N_t$ is at least $|N_t|\cdot f(k_{i_t})$.
Since not all items are allocated, we have $|N_t|\cdot f(k_{i_t}) < n$, which implies that $|N_t| < \frac{n}{f(k_{i_t})}$.
Therefore, the total number of agents that arrived before $i^*$ is less than
\begin{equation*}
    \sum_{t = 1}^\infty \frac{n}{f(k_{i_t})} = n\cdot \sum_{t = 1}^\infty \frac{1}{f(k_{i_t})} \leq n\cdot \sum_{r = 1}^\infty \frac{1}{f(r)} \leq \frac{n}{2},
\end{equation*}
where the first inequality follows since all $k_{i_1},k_{i_2},\ldots$ take different values, and the second inequality follows from Lemma~\ref{lemma:function-f}.
Therefore, we have $i^* \leq n/2$.

\smallskip

Consequently, we have $\ell_i = L = 2(\lceil\log m\rceil+1)$ for all $i\geq n/2$.
Recall that when agent $i$ arrives, among the $n$ bundles in the MMS partition, the algorithm assigns agent $i$ the $L$ bundles with the most unallocated items.
In Algorithm~\ref{alg:main-alg}, we use $R_i$ to denote the set of unallocated items at the beginning of round $i$ (when agent $i$ arrives).
Therefore we have
\begin{equation*}
    |R_{i+1}| = |R_i \setminus X_i| \leq \left(1-\frac{L}{n}\right)\cdot|R_i|,
\end{equation*}
which implies that
\begin{equation*}
    |R_{n+1}| \leq \left(1-\frac{L}{n}\right)^{n/2}\cdot|R_{\frac{n}{2}}| \leq m\cdot\exp\left(-\frac{L}{2}\right)= m\cdot\exp\left(-\lceil\log m\rceil-1\right) < 1,
\end{equation*}
which contradicts the assumption that some items are left unallocated (i.e., $|R_{n+1}|\geq 1$).
\end{proof}

\paragraph{Improved Ratio for Known $k$.} 
As discussed, one difficulty of the problem is that we do not know the number of types and our algorithm must dynamically update $\ell_i$ based on $k_i$.
If the total number of types $k$ is known to the algorithm upfront, then we can simply fix $\ell = \min\{n, k, L/2\}$ for all agents, and still ensure that all items are allocated.
The resulting allocation can be shown to be $\min\{n,k,\lceil \log m\rceil + 1\}$-MMS (see Appendix~\ref{sec:known-number-types}).

\subsection{Lower Bound on MMS Guarantee}

In the following, we present several lower bounds for the problem.
The first follows directly from the offline approximation of MMS under subadditive cost functions, while the second applies to online algorithms for binary additive functions.

\smallskip

As mentioned in the introduction, our online algorithm achieves a competitive ratio that matches the best known offline approximation ratio for MMS under subadditive cost functions, namely $O(\min\{n,\log m\})$, due to Li et al.~\cite{conf/nips/0037WZ23}.
They also established a lower bound of
\begin{equation*}
    \Omega\!\left(\min\left\{n,\frac{\log m}{\log\log m}\right\}\right),
\end{equation*}
showing that this approximation ratio is essentially optimal up to a $\log\log m$ factor.
Adapting their hard-instance construction, we further present (in Appendix~\ref{app:known-subadditive-lower}) an instance with $n$ agents whose cost functions belong to $k$ monotone submodular types, such that every allocation is $\Omega(k)$-MMS.
This implies that our online competitive ratio is also nearly optimal with respect to the parameter $k$.

\smallskip

We next establish lower bounds on the competitive ratio of deterministic online algorithms, which largely match the upper bounds obtained in our analysis.

\begin{theorem}
\label{thm:lower-bound-subadditive-uninformed}
No online algorithm can achieve a competitive ratio of $o(\min\{n, k\log k, \log m\})$, even for binary additive costs. 
\end{theorem}
\begin{proof}
We prove the following three statements, the combination of which yields the theorem:
\begin{itemize}
    \item No online algorithm is $(n-\epsilon)$-competitive for any constant $\epsilon>0$;
    \item No online algorithm is $(\frac{1}{4}\log m)$-competitive;
    \item No online algorithm is $O(k\log k)$-competitive.
\end{itemize}

We first present a general framework for constructing hard instances that rules out the existence of $\alpha$-competitive algorithms.
All instances in our construction are binary additive: each agent $i$ is associated with an approval set $A_i$, where $c_i(e)=1$ for every $e\in A_i$ and $c_i(e)=0$ otherwise.
Moreover, the approval sets are defined according to a nested structure.

\paragraph{Nested Structure.}
Consider $n$ agents arriving in the order $1,\dots,n$.
Let $R_i$ denote the set of unallocated items when agent $i$ arrives, with $R_1=M$, and let $r_i = |R_i|$.
The adversary defines the approval set of agent $i$ as $A_i = R_i$.
Thus, agent $i$ has cost $1$ for every remaining item and cost $0$ for every item allocated before her arrival.
It follows that $\MMS_i = \left\lceil \frac{r_i}{n} \right\rceil$.

\paragraph{Lower Bounds in $n$ and $m$.}
We first use this construction to derive lower bounds as functions of $n$ and $m$.
Suppose that an algorithm is $\alpha$-competitive.
Then, for every agent $i$, we have $|X_i| = c_i(X_i)
\le \alpha \cdot \left\lceil \frac{r_i}{n} \right\rceil$,
which implies
\begin{equation*}
r_{i+1}
= r_i - |X_i|
\ge r_i - \alpha \cdot \left\lceil \frac{r_i}{n} \right\rceil
\ge \left(1-\frac{\alpha}{n}\right)r_i - \alpha.
\end{equation*}
Iterating this recurrence yields
\begin{equation*}
r_n \ge \left(1-\frac{\alpha}{n}\right)^{n-1}m
- \alpha \cdot \sum_{j=0}^{n-2}
\left(1-\frac{\alpha}{n}\right)^j
= \left(1-\frac{\alpha}{n}\right)^{n-1}(m+n) -(n-1).
\end{equation*}
To rule out the existence of an $\alpha$-competitive algorithm, it suffices to construct an instance for which $r_n = \omega(n)$, since then agent $n$ has $\MMS_n = \omega(1)$ but has $c_n(X_n) = c_n(M)$.
\begin{itemize}
    \item If $\alpha = n-\epsilon$, then one can verify that for
    $m=(2n/\epsilon)^n$, we have $r_n = \omega(n)$.

    \item If $\alpha = \frac{1}{4}\log m$, then one can verify that for
    $m=2^{4n}$, we have $r_n = \omega(n)$.
\end{itemize}

\paragraph{Lower Bound in $k$.}
To establish the lower bound for instances with a bounded number of agent types, we use a similar construction, but agents of the same type arrive consecutively.
Suppose, for contradiction, that there exists a deterministic algorithm that guarantees an $O(k\log k)$-MMS allocation, where $k$ is the (unknown) number of agent types.
Let $f(t)$ denote the competitive ratio achieved by the algorithm when the number of agent types is $t$.
By the definition of big-$O$ notation, there exist constants $a$ and $b$ such that, for all $t\ge a$, $f(t)\leq b\cdot t\cdot \log t$.
The first type of agents has approval set $M$.
We let $\frac{n}{2f(1)}$ agents\footnote{For sufficiently large $n$, we may assume that $n/(2f(t))$ is an integer for every $t\in\{1,\ldots,k\}$. Similarly, for sufficiently large $m$, we may assume that $m/n$ is an integer.}
of this type arrive.
Since at the moment all agents belong to one type, the algorithm cannot allocate more than $f(1)\cdot \MMS_1 = f(1)\cdot \frac{m}{n}$ items to any agent.
Therefore, the number of remaining items after these agents arrive is at least
\begin{equation*}
m -
f(1)\cdot \frac{m}{n}\cdot \frac{n}{2f(1)}
= \frac{m}{2}.
\end{equation*}
Next, we define the approval set of the second type to be the set of remaining items and let $\frac{n}{2f(2)}$ agents of this type arrive.
By the same argument, the number of remaining items after these agents have been allocated is at least $m/4$.
Continuing in this manner, after introducing $t$ types, the number of remaining items is at least ${m}/{2^t}$.
We terminate the construction once the total number of arrived agents reaches $n$.
By choosing $m$ sufficiently large, we can ensure that at least $n$ items remain when the process terminates.

It therefore remains to show that there exists some $t$ such that $\sum_{i=1}^{t}\frac{n}{2f(i)} > n$,
or equivalently,
$\sum_{i=1}^{t}\frac{1}{f(i)} > 2$.
This condition is guaranteed whenever the infinite series $\sum_{i=1}^{\infty}\frac{1}{f(i)}$ diverges.
Since $f(i)\le b\cdot i\cdot \log i$ for all $i\ge a$, we have
\begin{equation*}
\sum_{i=1}^{\infty}\frac{1}{f(i)}
\ge
\sum_{i=a}^{\infty}\frac{1}{f(i)}
\ge
\sum_{i=a}^{\infty}\frac{1}{b\cdot i\cdot \log i}
=
\infty.
\end{equation*}

Hence, there exists a finite $t$ for which $\sum_{i=1}^{t}\frac{1}{f(i)} > 2$.
Using this value of $t$ in the above construction yields a hard instance, which shows that no deterministic algorithm can guarantee an $O(k\log k)$ competitive ratio.
\end{proof}

\section{Known Additive Cost Function Types}
\label{sec:known-types-additive}

In this section, we study the known-types setting, where the algorithm is given a collection of $k$ cost functions $c_1,\ldots,c_k$, each corresponding to a distinct agent type. Each agent reveals her type only upon arrival. We do not assume that every type appears at least once, and it is possible that $k>n$.
As discussed in the previous section, the offline inapproximability result leaves very little room for improving the competitive ratio for subadditive costs. Therefore, in this section we focus on additive cost functions.
We use $\MMS_t$ to denote the MMS value of the cost function $c_t$.
We normalize the cost functions so that $\MMS_t = 1$ for all $t$.
This further implies $c_t(M)\leq n$ and $c_t(e)\le 1$ for every item $e\in M$ and every $t\in[k]$.

\subsection{From Online Allocation to Offline Universal Partitions}

In this subsection, we reduce the design of competitive online algorithms to the offline problem of computing approximate universal partitions.

\begin{definition}[$\alpha$-MMS Universal Partition]
\label{def:universal-partition}
Given cost functions $c_1,\dots,c_k$, an $n$-partition $\bP=(P_1,\dots,P_n)$ of $M$ is called an $\alpha$-MMS universal partition if
\begin{equation*}
c_t(P_j)\le \alpha\cdot \MMS_t,
\qquad
\text{for every } t\in[k] \text{ and every } j\in[n].
\end{equation*}
\end{definition}

The following lemma shows that the existence of an $\alpha$-MMS universal partition immediately yields an $\alpha$-competitive online allocation.

\begin{lemma}
\label{lem:universal-partition-online-reduction}
Given an $\alpha$-MMS universal partition, assigning a distinct bundle to each arriving agent yields an $\alpha$-competitive online algorithm.
\end{lemma}

\begin{proof}
Let $\bP=(P_1,\ldots,P_n)$ be an $\alpha$-MMS universal partition, and we allocate bundle $P_i$ to the $i$th arriving agent. Although the type of agent $i$ is unknown before arrival, the defining property of the universal partition ensures that every bundle has cost at most $\alpha\cdot \MMS_t$ under every cost function type $t$. Consequently, the allocation is $\alpha$-MMS.
\end{proof}

Therefore, it suffices to focus on computing an offline universal partition of the items, without explicitly considering the online allocation process.

In the remainder of this section, we present two algorithms that compute $O(\frac{\log(kn)}{\log\log(kn)})$-MMS and $O(\log k)$-MMS universal partitions, respectively. We then construct a binary additive instance for which every universal partition has an approximation ratio of $\Omega(\frac{\log k}{\log\log k})$.

\subsection{An \texorpdfstring{$O(\frac{\log(kn)}{\log\log(kn)})$}{}-MMS Universal Partition}
\label{sec:known-types-additive-logkn}

Our first algorithm for computing the universal partition is based on a simple balls-into-bins paradigm and its standard probabilistic analysis.
Consider the process in which each of the $m$ items is assigned independently and uniformly at random to one of the bundles.
Using the fact that $\MMS_t \geq \max_{e\in M} c_t(e)$, together with appropriate concentration inequalities, one can show that, with high probability, every bundle has cost at most $O\!\left(\frac{\log(kn)}{\log\log(kn)}\right)$.
Consequently, the existence of such a universal partition follows rather directly from the probabilistic argument.

To obtain a deterministic polynomial-time algorithm for computing such a partition, however, a more careful approach is required.
In particular, we define a suitable potential function and assign the items sequentially so as to maintain a bounded potential throughout the process.

The main result of this section is summarized below.

\begin{lemma} 
\label{lem:additive-logkn-universal} 
Given $k$ cost functions, we can compute an $O\!\left(\frac{\log(kn)}{\log\log(kn)}\right)$-MMS universal partition in polynomial time deterministically. 
\end{lemma}
\begin{proof}
We first show that randomly allocating the items to $n$ bundles yields the desired universal partition with positive probability. 
To establish the existence of such a partition under random allocation, we define the parameter
\begin{equation*}
B:= \left\lceil \frac{4\cdot \log(kn)}{\log\log(kn)} \right\rceil
\end{equation*}
and assume that $B = \omega(1)$ (otherwise, $k$ and $n$ are both constants, and the lemma holds trivially since all partitions are $n$-MMS). Our goal is to find a partition such that each bundle of the partition has a cost of at most $B$ for every type. 
Given a partition $\bP = (P_1,\ldots,P_n)$ of $M$, we define its potential as
\begin{equation*}
    \Phi(\bP) := \sum_{t=1}^{k}\sum_{j=1}^{n} B^{c_t(P_j)-B}.
\end{equation*}

Note that if there exists a bundle $P_i$ and a type $t$ such that $c_t(P_i) > B$, then we have 
\begin{equation*}
    \Phi(\bP) \geq B^{c_t(P_i) - B} > 1.
\end{equation*}

Therefore, we conclude that a partition $\bP$ is a $B$-MMS universal partition if $\Phi(\bP) \leq 1$.

\smallskip

We now bound the expectation of this potential for the random partition. 
For any item $e \in M$ and bundle $P_j$, we use $\mathbf{1} [e\in P_j]$ to denote the indicator of the event that $e$ is assigned to $P_j$. 
Note that for every fixed $j$, the indicators $\{\mathbf{1}[e\in P_j]\}_{e\in M}$ are independent random variables and satisfy $\Pr[\mathbf{1}[e\in P_j] = 1] = 1/n$. For any type $t\in[k]$, since $c_t$ is additive, we have
$c_t(P_j)=\sum_{e\in M} c_t(e) \cdot \mathbf{1}[e\in P_j]$.
Therefore,
\begin{equation*}
B^{c_t(P_j)} = B^{\sum_{e\in M} c_t(e) \cdot \mathbf{1}[e\in P_j]} = \prod_{e\in M} B^{c_t(e)\cdot \mathbf{1}[e\in P_j]}.    
\end{equation*}

Taking expectation and using the independence of the indicators, we get
\begin{align*}
\mathbb E\!\left[B^{c_t(P_j)}\right] &= \prod_{e\in M} \mathbb E\!\left[B^{c_t(e)\cdot \mathbf{1}[e\in P_j]}\right] 
= \prod_{e\in M} \left( \left(1-\frac1n\right)\cdot B^0 + \frac1n\cdot B^{c_t(e)} \right)  \\
&= \prod_{e\in M} \left(1-\frac1n+\frac1n \cdot B^{c_t(e)}\right).
\end{align*}

Recall that we have $c_t(e)\in [0, 1]$ and $c_t(M) \leq n$. Therefore, we derive
\begin{align*}
\mathbb E\!\left[B^{c_t(P_j)}\right]
&\le
\prod_{e\in M}
\left(1+\frac{B-1}{n}\cdot c_t(e)\right)  
\le
\prod_{e\in M}
\exp\left(\frac{B-1}{n}\cdot c_t(e)\right)  \\
&=
\exp\left(\frac{B-1}{n}\cdot \sum_{e\in M}c_t(e)\right) 
=
\exp\left(\frac{B-1}{n}\cdot c_t(M)\right) 
\le
e^{B-1},
\end{align*}
where the first inequality holds by $B^{c_t(e)} \leq 1+ (B-1)\cdot c_t(e)$, and the second inequality uses $1+x\le e^x$ for all $x$.
By linearity of expectation, we have
\begin{align*}
\mathbb E[\Phi(\mathbf P)]
&=
\sum_{t=1}^{k}\sum_{j=1}^{n}
B^{-B}\cdot \mathbb{E}\left[B^{c_t(P_j)}\right]  
\le
kn\cdot B^{-B}\cdot e^{B-1}  \\
&=
kn\cdot \exp\bigl(-B(\log B-1)-1\bigr) \leq kn \cdot \exp(-2\log(kn)-1) < 1,
\end{align*}
where the second inequality holds by $B(\log B - 1) > 2\log (kn)$.
This establishes the existence of a $B$-MMS universal partition.
We now show how to compute such a partition deterministically in polynomial time by derandomizing the random allocation process via conditional expectations.

Fix an arbitrary ordering of the items. We allocate the items to bundles one at a time while maintaining an upper bound of $1$ on the conditional expectation of the potential. Specifically, given a partial allocation $\bX=(X_1,X_2,\ldots,X_n)$, let $\E[\Phi(\bP)\mid \bX]$ denote the conditional expectation of the potential of $\bP$ given that $X_i\subseteq P_i$ for every $i$. The expectation is taken over the random assignment of the remaining unallocated items, each of which is assigned independently and uniformly at random to a bundle. By the analysis above, we have
\begin{equation*}
    \E[\Phi(\bP)\mid (\varnothing,\ldots,\varnothing) ] < 1.
\end{equation*}
Moreover, when $\bX$ is a complete allocation, we have
\begin{equation*}
    \E[\Phi(\bP)\mid \bX]=\Phi(\bX).
\end{equation*}

Our goal is to maintain the invariant that $\E[\Phi(\bP)\mid \bX] < 1$
throughout the allocation process. Since the invariant holds initially when no items have been allocated, maintaining it until the end guarantees that the resulting complete allocation is a $B$-MMS universal partition.

Suppose that the current partial allocation is $\bX$, and let $e$ be the next unallocated item. For each bundle $P_j$, define $\Psi_j$ as the conditional expectation of $\Phi(\bP)$ given $\bX$ and the additional condition that $e$ is assigned to $P_j$, i.e.,
\begin{equation*}
    \Psi_j :=
    \E\!\left[\Phi(\bP)\,\middle|\, (X_1,\ldots,X_j\cup\{e\},\ldots,X_n)\right].
\end{equation*}

Since $e$ is assigned uniformly at random to one of the $n$ bundles in $\E[\Phi(\bP)\mid \bX]$, we have
\begin{equation*}
    \E[\Phi(\bP)\mid \bX]
    = \frac{1}{n} \cdot \sum_{j=1}^{n}\Psi_j.
\end{equation*}

Therefore, there exists some $j\in[n]$ such that $\Psi_j \leq \E[\Phi(\bP)\mid \bX] < 1$.
We assign $e$ to such a bundle $P_j$, thereby preserving the invariant, and continue until all items have been allocated.

It remains to show that, given any (partial) allocation $\bX$, the conditional expectation $\E[\Phi(\bP)\mid \bX]$ can be computed in polynomial time. By the definition of the potential function, it suffices to compute $\E[B^{c_t(P_j)}\mid \bX]$ for every type $t$ and bundle $P_j$. Using the same argument as above, we obtain
\begin{equation*}
    \E\left[B^{c_t(P_j)}\mid \bX\right]
    = B^{c_t(X_j)}\cdot \prod_{e\notin \cup_{i\in N} X_i}
    \left(1-\frac{1}{n} +\frac{1}{n}\cdot B^{c_t(e)} \right).
\end{equation*}

This expression is clearly computable in polynomial time. Summing these quantities over all types $t$ and all bundles $P_j$ yields the conditional expectation of the potential function.
\end{proof}

\subsection{\texorpdfstring{$O(\log k)$}{}-MMS Universal Partition}
\label{sec:known-types-additive-logk}

Next, we present an alternative algorithm for computing an $O(\log k)$-MMS universal partition. Unlike the algorithm from the previous section, its approximation guarantee is independent of $n$, and therefore improves upon the previous result when $k \ll n$.
As before, we begin by describing a randomized process that produces the desired partition with positive probability, and then show how to derandomize it using a potential-function-based argument.
At a high level, the limitation of the previous randomized approach is that it constructs all $n$ bundles simultaneously. As a result, the analysis incurs a $\log(kn)$ factor, since it (implicitly) requires a union bound over all $n$ bundles and all $k$ cost functions. To remove the dependence on $n$, we instead adopt a progressive procedure that constructs the bundles sequentially, one at a time.
The key idea is that, at each step, it suffices to show that both the newly constructed bundle and the set of remaining unallocated items have bounded cost under each of the $k$ cost functions with positive probability. Consequently, the analysis only requires a union bound over the $k$ cost functions, thereby eliminating the dependence on the number of bundles $n$ from the approximation guarantee.
As a warm-up, we first present a simple algorithm for computing a $(k+1)$-MMS universal partition.

Fix an arbitrary ordering of the items and add them to an initially empty bundle $S$ until $\sum_{t\in[k]} c_t(S) > k$, or until no items remain.
Recall that $\MMS_t = 1$ for every type $t\in [k]$, and that $c_t(e)\le 1$ for every type $t\in [k]$ and item $e\in M$.
Therefore, for the resulting bundle, we have $c_t(S)\le k+1$ for every $t \in [k]$, implying that $S$ is $(k+1)$-MMS for all types.
We then remove $S$ from consideration and repeat the same procedure on the remaining items to construct the next bundle, continuing until $n$ bundles have been produced.
Since $\sum_{t\in[k]} c_t(M) \le k n$, the process must terminate within $n$ rounds.
Hence, the resulting partition is a $(k+1)$-MMS universal partition.
Our improved algorithm follows the same high-level principle, but is significantly more careful in selecting the items allocated to each bundle $S$.

In the remainder of this section, we assume that $k=\omega(1)$; otherwise, the above $(k+1)$-MMS universal partition already yields an $O(\log k)$-MMS guarantee.
The algorithm proceeds in $n$ rounds, with one bundle constructed in each round.
Let $S_1,S_2,\ldots,S_{i-1}$ denote the bundles constructed before round $i$, and let
$R_i = M \setminus \bigcup_{j<i} S_j$
be the set of remaining items.

We maintain the following invariant throughout the execution of the algorithm.

\begin{invariant}
\label{invariant:logk_universal}
    At the beginning of round $i$, for all types $t\in [k]$, we have
    \begin{align*}
        c_t(S_j) \leq 12\cdot \log k, \quad \forall j < i, 
        \text{ and } \quad
        c_t(R_i) \leq (n-i+1)\cdot \log k.
    \end{align*}
\end{invariant}

Clearly, the invariant holds at the beginning of round $1$, as $c_t(M)\leq n$ holds for all types $t$.
To maintain the invariant, it remains to show that if the invariant holds at the beginning of round $i$, then it also holds at the end of the round.

\begin{lemma}
\label{lem:additive-logk-extraction}
Suppose Invariant~\ref{invariant:logk_universal} holds at the beginning of round $i$.
Then we can find a set $S\subseteq R_i$ satisfying $c_t(S) \leq 12\cdot \log k$ and $c_t(R_i \setminus S)\le (n-i)\cdot \log k$ for every $t\in[k]$. Moreover, such a set can be found deterministically in polynomial time.
\end{lemma}
\begin{proof}
If $n-i+1 \leq 12$, then by Invariant~\ref{invariant:logk_universal}, choosing $S = R_i$ yields the lemma.
Therefore, we assume that $n-i+1 > 12$.
We construct $S$ by the following random sampling process: for each item in $R_i$, we include it in $S$ independently with probability $p := \frac{6}{n-i+1} < \frac{1}{2}$.
We prove that with positive probability, for all $t$ we have $c_t(S) \leq 12\cdot \log k$ and $c_t(R_i\setminus S) \leq (n-i)\cdot \log k$.

We first identify a group of \emph{critical} types, which is defined as
\begin{equation*}
    A:=\{t\in[k]: c_t(R_i) > (n-i)\cdot \log k\}.
\end{equation*}

For a type $t\notin A$, we have $c_t(R_i\setminus S)\le c_t(R_i)\le (n-i)\cdot \log k$.
Therefore for any $S$, the set $R_i\setminus S$ of remaining items satisfies Invariant~\ref{invariant:logk_universal}.
For a type $t \in A$, if we can ensure that $c_t(S) \ge \log k$, the set $R_i\setminus S$ also satisfies Invariant~\ref{invariant:logk_universal}.
Therefore, it suffices to show that with positive probability, we have $c_t(S) \leq 12\cdot \log k$ for all $t\in [k]$ and $c_t(S) \ge \log k$ for all $t\in A$.

We define the potential of $S$ as
\begin{equation*}
    \Phi(S) := \sum_{t\in [k]} 2^{c_t(S)-12\log k} + \sum_{t\in A} 5^{\log k-c_t(S)}.
\end{equation*}

Therefore, $\Phi(S)\le 1$ implies
\begin{equation*}
    c_t(S)\le 12\cdot \log k
    \quad\text{for every }t\in[k],
    \quad \text{ and } \quad
    c_t(S)\ge \log k
    \quad\text{for every }t\in A.
\end{equation*}

Next, we show that $\E[\Phi(S)] < 1$, which shows the existence of the desired $S$.
We have
\begin{align*}
\label{eq:expectation-of-potential-logk}
    \E[\Phi(S)] &= \E\left[\sum_{t\in [k]} 2^{c_t(S)-12\log k} + \sum_{t\in A} 5^{\log k-c_t(S)}\right] \\
    &= \E\left[2^{-12\log k}\cdot\sum_{t\in [k]} 2^{c_t(S)} + 5^{\log k}\cdot\sum_{t\in A} 5^{-c_t(S)}\right] \\
    &= 2^{-12\log k}\cdot \sum_{t\in [k]} \E\left[2^{c_t(S)}\right] + 5^{\log k}\cdot \sum_{t\in A}\E\left[5^{-c_t(S)}\right].
\end{align*}
        
We first bound the term $\E[2^{c_t(S)}]$. By a similar argument as in Section~\ref{sec:known-types-additive-logkn}, we have
\begin{align*}
\mathbb E\left[2^{c_t(S)}\right] & = \mathbb E\left[ \prod_{e\in R_i}2^{c_t(e)\cdot\mathbf{1}[e\in S]}\right]
=\prod_{e\in R_i}\mathbb E\left[2^{c_t(e)\cdot\mathbf{1}[e\in S]}\right] 
=\prod_{e\in R_i}\left(1-p+p\cdot 2^{c_t(e)}\right) \\
& \le \prod_{e\in R_i} \left(1+p\cdot c_t(e)\right)
\le \prod_{e\in R_i} e^{p\cdot c_t(e)}
=  \exp\left(p\cdot c_t(R_i) \right)
\le \exp\left(6\log k\right),
\end{align*}
where the last inequality follows by $p=6/(n-i+1)$ and $c_t(R_i)\le (n-i+1)\log k$. 

Next, fix a critical type $t\in A$ (which satisfies $c_t(R_i) > (n-i)\cdot \log k$). Since $c_t(e)\in[0,1]$, we have $5^{-c_t(e)}\le 1-\frac{4}{5}\cdot c_t(e)$. By the same argument as above, we derive
\begin{align*}
\mathbb E\left[5^{-c_t(S)}\right] & = \prod_{e\in R_i} \left(1-p+p\cdot 5^{-c_t(e)}\right)
\le \prod_{e\in R_i} \left(1-\frac{4p}{5}\cdot c_t(e)\right) 
\le \exp\left(-\frac{4p}{5}\cdot c_t(R_i)\right) \\
& \leq \exp\left(- \frac{4}{5}\cdot\frac{6}{n-i+1} \cdot (n-i)\cdot \log k\right)
< \exp\left( -4 \log k \right),
\end{align*}
where the last inequality uses $n-i+1>12$.
Therefore, we have
\begin{equation*}
\E [\Phi(S)] \leq  2^{-12\log k}\cdot k\cdot e^{6\log k} + 5^{\log k} \cdot k \cdot e^{-4\log k} 
\leq k^{-2}\cdot k + k^{-2}\cdot k = \frac{2}{k} < 1.
\end{equation*}

By a similar argument as used in Section~\ref{sec:known-types-additive-logkn}, we can derandomize the above process by maintaining a bounded conditional expectation on $\E [\Phi(S)]$, and thus compute the set $S$ deterministically in polynomial time.
We defer the details to Appendix~\ref{app:additive-derandomization}. 
\end{proof}

By Lemma~\ref{lem:additive-logk-extraction}, if Invariant~\ref{invariant:logk_universal} holds when round $i$ begins, then by setting $S_i = S$, we can ensure that Invariant~\ref{invariant:logk_universal} holds when round $i$ ends.
Notice that for $i=n$, Invariant~\ref{invariant:logk_universal} implies that $c_t(R_n) \leq \log k$, and it suffices to set $X_n = R_n$.
Therefore, when the algorithm terminates, all items are allocated, and we obtain a $(12\cdot \log k)$-MMS universal partition.

\begin{lemma}
\label{lem:additive-logk-universal}
Given $k$ cost functions, we can compute an $O(\log k)$-MMS universal partition in polynomial time deterministically.
\end{lemma}

Together with Lemma~\ref{lem:universal-partition-online-reduction} and Lemma~\ref{lem:additive-logkn-universal}, we have our second main result.

\begin{theorem}
\label{thm:known-types-additive-upper}
Given $k$ known types of additive cost functions, there is a deterministic polynomial-time online algorithm with a competitive ratio of
\begin{equation*}
    O\!\left(\min\left\{\frac{\log(kn)}{\log\log(kn)},\log k\right\}\right).
\end{equation*}
\end{theorem}

\subsection{Lower Bound for Universal Partitions}
\label{sec:known-types-additive-lower}
In this subsection, we establish a lower bound for approximate universal partitions. We show that even for binary additive cost functions, there are instances for which every universal partition must incur an approximation ratio of
$\Omega\!\left(\frac{\log k}{\log\log k}\right)$.

\begin{theorem}
\label{thm:additive-universal-lower}
There exists an instance with $k$ binary additive cost functions such that every universal partition has an approximation ratio of
$\Omega\!\left(\frac{\log k}{\log\log k}\right)$.
\end{theorem}

\begin{proof}
Fix any integer $q \ge 2$.
We show that there exists an instance with $n=(2q)^{2q}$ agents, $m=q\cdot n$ items, and $k=(2q)^{q+3}$ binary additive cost functions $c_1, c_2, \ldots, c_k$ such that, for every partition $(P_1,P_2,\ldots,P_n)$ of $M$, there exists a cost function $c_t$ and a bundle $P_i$ satisfying $c_t(P_i)\ge q$.
Our construction ensures that $\MMS_t=1$ for all $t$, which implies that every universal partition has an approximation ratio of at least
$q=\Omega\!\left(\frac{\log k}{\log\log k}\right)$.

We construct $k$ binary additive cost functions $c_1,\ldots,c_k$.
Let $A_t$ denote the approval set of $c_t$, that is,
$c_t(e)=1$ if $e\in A_t$ and $c_t(e)=0$ otherwise.
Let the size of every approval set be $n$. 
This implies that $\MMS_t=1$ for every $t\in[k]$.

\begin{claim}
\label{claim:A_covers_Q}
There exist size-$n$ sets $A_1,\ldots,A_k$ such that for every size-$q$ subset $Q\subseteq M$, there exists $t\in[k]$ with $Q\subseteq A_t$.
\end{claim}
\begin{proof}
Consider the following random construction.
For each $t\in[k]$, choose $A_t$ independently and uniformly at random from all size-$n$ subsets of $M$.

Fix a size-$q$ subset $Q\subseteq M$ and some $t\in[k]$. 
Over the randomness of $A_t$, we have
\begin{equation*}
\Pr[Q\subseteq A_t]
=\frac{\binom{m-q}{n-q}}{\binom{m}{n}}
=\prod_{r=0}^{q-1}\frac{n-r}{m-r}
\ge \left(\frac{n-q}{m}\right)^q
\ge \frac{1}{(2q)^q},
\end{equation*}
where the last inequality uses $m=q\cdot n$ and $n=(2q)^{2q}$.

Therefore, the probability that $Q$ is not contained in any of
$A_1,\ldots,A_k$ is at most
\begin{equation*}
\left(1-\frac{1}{(2q)^q}\right)^k
\le \exp\left(-\frac{k}{(2q)^q}\right)
= \exp\left( -(2q)^3 \right)
<
(2q)^{-(2q)^2},
\end{equation*}
where the last inequality follows from the fact that
$2q>\log(2q)$ for every $q\ge 2$.

Since there are at most
\begin{equation*}
    \binom{m}{q}\le m^q = \bigl(q\cdot(2q)^{2q}\bigr)^q
\end{equation*}
choices for $Q$, a union bound implies that the probability that some size-$q$ subset is not contained in any of
$A_1,\ldots,A_k$ is at most
\begin{equation*}
\left(q\cdot(2q)^{2q}\right)^q
\cdot(2q)^{-(2q)^2}
= q^q\cdot(2q)^{-2q^2} < 1.
\end{equation*}

Hence, with positive probability, every size-$q$ subset of $M$ is contained in at least one of the sets
$A_1,\ldots,A_k$.
This proves the existence of the desired collection of sets.
\end{proof}

Now consider an arbitrary partition $\bP$ of $M$.
Since $m=q\cdot n$, at least one bundle $P_i$ must contain at least $q$ items.
Choose an arbitrary size-$q$ subset $Q\subseteq P_i$.
By Claim~\ref{claim:A_covers_Q}, there exists some $t\in[k]$ such that
$Q\subseteq A_t$.
Consequently, $c_t(P_i)\ge c_t(Q)=q$.

Since $\MMS_t=1$, the approximation ratio with respect to agent $t$ is at least $q$.
Therefore, every universal partition has an approximation ratio of at least
$q=\Omega\!\left(\frac{\log k}{\log\log k}\right)$.
\end{proof}

We remark that this lower bound applies only to universal partitions and does not directly imply a lower bound for general online algorithms.

\section{Known Binary Additive Cost Function Types}
\label{sec:known-types-binary}

In this section, we turn our attention to binary additive costs, where each agent $i$'s preference is completely determined by her approval set $A_i$. That is, $c_i(e)=1$ for every $e\in A_i$ and $c_i(e)=0$ otherwise.
The algorithm is given a collection of $k$ binary cost function types, not all of which are necessarily realized, and $k$ may be larger than $n$.
Since there are $2^m$ possible subsets of items, the number of distinct binary types is bounded by $2^m$.
Consequently, the known-types setting with $k=2^m$ is equivalent to the unknown-types setting.
As a direct consequence of Theorem~\ref{thm:lower-bound-subadditive-uninformed} (no online algorithm is $o(\log m)$-competitive), no deterministic online algorithm can achieve an $o(\log\log k)$-MMS guarantee for binary additive costs when $k$ is unrestricted\footnote{Note that binary additive costs are a subclass of additive costs, this lower bound also applies to general additive costs. This is the best known lower bound for general online algorithms with known additive types.}
We therefore focus on the regime $k\le n$ and seek constant-competitive algorithms.

The binary structure allows the online algorithm to reason primarily about bundle sizes rather than their costs.
Consider the following natural approach.
Suppose the goal is to compute an $\alpha$-MMS allocation for some constant $\alpha$.
Upon arrival, the algorithm learns the approval set $A_i$ of agent $i$.
Since every item in $M\setminus A_i$ incurs zero cost to agent $i$, the algorithm can first allocate all currently unallocated items in $M\setminus A_i$ to agent $i$.
Among the remaining unallocated items in $A_i$, it then allocates a bundle of size $\alpha\cdot\MMS_i$.
The key question is how to choose this bundle.
A natural option is to select a random bundle of size $\alpha\cdot\MMS_i$.
Intuitively, for a sufficiently large constant $\alpha$, one might expect all items to be allocated by the end with high probability.
However, since the allocation events are highly dependent, it is difficult to characterize the intricate interactions of the events across different rounds.
Moreover, the previous hardness results show that a constant competitive ratio cannot be achieved without exploiting the knowledge of the $k$ cost function types.

Motivated by these observations, we extend the idea of approximate universal partitions to the binary setting, while constructing bundles on-the-fly as agents arrive.
The main advantage is that we can relax the requirements in the definition of an approximate universal partition.
Specifically, when agent $i$ arrives, we compute a bundle $X_i$ and only require that $c_i(X_i)$ be appropriately bounded; we do not need to control the cost of $X_i$ under other cost function types.
To ensure that the allocation is ``safe'' for future agents, however, we maintain the invariant that the set of remaining items has bounded cost under every cost function type.
We refer to this property as maintaining a \emph{Universal Residual}.

As in the previous sections, we first show that there exists a randomized process that maintains the universal residual with positive probability in every round, and then derandomize this process in polynomial time.
As a result, we obtain a $3$-competitive algorithm (Theorem~\ref{thm:known-binary-three-mms}) whenever $k\le n$.
To complement this positive result, we also prove that no algorithm can achieve an approximation ratio better than $2$ even when $k$ is a constant (Theorem~\ref{thm:binary-small-k-two-lower}).

\subsection{The Algorithm and the Invariant}
\label{sec:known-binary-three-mms}

We focus on the setting where $k \le n$ and propose a $3$-competitive algorithm.
Let $c_1,c_2,\ldots,c_k$ be the known cost function types, where each function $c_t$ has approval set $A_t\subseteq M$.
Without loss of generality, we assume that $A_t \neq \varnothing$ and $\MMS_t \geq 1$ for all $t$; otherwise, we can remove $c_t$ from the collection of types, and if an agent of type $t$ arrives, the algorithm can terminate immediately by allocating all items to this agent.
For ease of notation, we write $\mu_t = \MMS_t$ in the remaining sections.
For any two integers $x$ and $y$, we use $\binom{x}{y}$ to denote the binomial coefficient, with the standard convention that $\binom{x}{y} = 0$ whenever $y > x$.

Following the previous discussion, we present our algorithm.
In round $i$, when agent $i$ arrives, we compute a bundle $X_i$ and allocate it to agent $i$.
Let $R_i$ be the set of unallocated items at the beginning of round $i$, with $R_1 = M$.
We maintain the following invariant.

\begin{invariant}[Universal Residual]
\label{invariant:binary-3mms-potential}
At the beginning of round $i$, for every type $t\in [k]$, we have
\begin{equation*}
|R_i\cap A_t| \leq 2(n-i+1)\cdot \mu_t.
\end{equation*}
\end{invariant}

In particular, the above invariant implies that when agent $n$ arrives, the total cost of the unallocated items under her cost function is at most $2\cdot \mu_t$, and hence we can safely allocate all remaining items to agent $n$.
Therefore, it suffices to show that the invariant can be maintained while ensuring that each online agent $i$ receives a bundle of cost at most $3$ times her MMS.
% Note that the invariant holds in the first round, since $c_t(M) \leq n\cdot \mu_t$ for all $t$.
% 
We introduce the following potential function and relate the maintenance of the invariant to upper bounding this potential.
For a set $R$ of unallocated items and $r$ remaining agents, define
\begin{equation*}
\Phi(R,r) = \sum_{t\in [k]}
\sum_{\ell\in[r]}
\binom{|R\cap A_t|}{2\ell \cdot \mu_t+1}
\left( \prod_{j=\ell+1}^{r}
\left(1-\frac{3}{2j}\right) \right)^{2\ell \cdot \mu_t+1}.
\end{equation*}

\begin{lemma}
\label{lem:binary-3mms-residual-safety}
If $\Phi(R,r) < 1$, then for every type $t\in[k]$, we have $|R\cap A_t|\le 2r \cdot \mu_t$.
\end{lemma}
\begin{proof}
Fix an arbitrary $t\in [k]$.
The term corresponding to $\ell=r$ in $\Phi(R,r)$ is
$$
\binom{|R\cap A_t|}{2r \cdot \mu_t+1}
\left(
\prod_{j=r+1}^{r}\left(1-\frac{3}{2j}\right)
\right)^{2r \cdot \mu_t+1}
= \binom{|R\cap A_t|}{2r \cdot \mu_t+1}.
$$
If $|R\cap A_t|\ge 2r\cdot \mu_t+1$, this term is at least $1$.
This would imply $\Phi(R,r)\ge 1$, contradicting the assumption.
Hence $|R\cap A_t|\le 2r \cdot \mu_t$.
\end{proof}

By the above lemma, to maintain Invariant~\ref{invariant:binary-3mms-potential}, it suffices to ensure that $\Phi(R_i, n-i+1) < 1$ holds for all $i$.
To prove this, we show that if $\Phi(R_i, n-i+1) < 1$ holds for some $i< n$, then by choosing an appropriate set of items $X_i$ to allocate to agent $i$, we also have $\Phi(R_{i+1}, n-i) < 1$; see Algorithm~\ref{alg:known-binary-three-mms} for the details.

\begin{algorithm}[!htb]
\caption{$3$-MMS allocations for binary additive costs} \label{alg:known-binary-three-mms}
\KwIn{A set $M$ of items, the number of online agents $n$, cost functions with approval sets $A_1,\ldots,A_k$}
$R_1\gets M$\;
\For{each online agent $i$, where $i=1,\ldots,n$}{
\If{$i = n$}{
Let $X_i\gets R_i$ and $R_{n+1}\gets\varnothing$;
}
\Else{
Let $t_i$ be the type of agent $i$, and $B_i\gets R_i\cap A_{t_i}$\;
Find a set $S\subseteq B_i$ with $|S|=\left\lceil\frac{3|B_i|}{2(n-i+1)}\right\rceil$ and $\Phi(B_i\setminus S, n-i) \le \Phi(R_i,n-i+1)$\;
\tcp{Existence follows by Lemma~\ref{lem:binary-3mms-one-step}}
Let $X_i\gets (R_i \setminus B_i)\cup S$ and
$R_{i+1} \gets B_i \setminus S$;
}
}
\KwOut{Allocation $\mathbf X=(X_1,\dots,X_n)$}
\end{algorithm}

We first show that $\Phi(R_1,n) < 1$. Note that $R_1 = M$.

\begin{lemma}
\label{lem:binary-3mms-initial-potential}
We have $\Phi(M, n) < 1$.
\end{lemma}
\begin{proof}
Fix any type $t\in [k]$. Since $\mu_t=\lceil |A_t|/n\rceil$, we have $|A_t|\le n\cdot\mu_t$.
Then
$$
\prod_{j=\ell+1}^{n}\left(1-\frac{3}{2j}\right)
\le
\left(\frac{\ell}{n}\right)^{3/2}
\le
\left(\frac{2\ell \cdot \mu_t+1}{2n\cdot\mu_t}\right)^{3/2},
$$
where the first inequality follows from the following claim, whose proof is given in Appendix~\ref{app:omitted-3mms}.

\begin{claim} \label{claim:binary-3mms-product-bounds}
For all $1\le \ell\le r$, we have $\prod_{i=\ell+1}^{r}\left(1-\frac{3}{2i}\right)
\le \left(\frac{\ell}{r}\right)^{3/2}$.
\end{claim}

Therefore, the contribution of type $t$ to $\Phi(M,n)$ is at most
$$
\sum_{\ell=1}^{n} \binom{|A_t|}{2\ell \cdot \mu_t+1} \left( \frac{2\ell \cdot \mu_t+1}{2n \cdot \mu_t} \right)^{3(2\ell \cdot \mu_t+1)/2} < \frac{1}{n \cdot \mu_t} \leq \frac{1}{n},
$$
where the first inequality follows from the fact that $|A_t| \leq n\cdot\mu_t$, together with the following claim, whose proof is also deferred to Appendix~\ref{app:omitted-3mms}, applied with $N=n \cdot \mu_t$; the second inequality follows since $\mu_t \geq 1$.

\begin{claim} \label{claim:binary-3mms-numerical}
For every integer $N\ge 1$, we have
\begin{equation*}
    \sum_{i=3}^{N} \binom{N}{i} \left(\frac{i}{2N}\right)^{3i/2} < \frac{1}{N}.
\end{equation*}
\end{claim}

Since there are at most $k \le n$ types in total, summing these potentials over all types yields the lemma.

Note that this is the only place where the condition $k \leq n$ is required.
Therefore, our result also holds under the slightly weaker assumption that $\sum_{t\in [k]} (1/\mu_t) \leq n$.

\end{proof}

Given that Invariant~\ref{invariant:binary-3mms-potential} holds at the beginning of round $i$, the bundle $X_i$ allocated to agent $i$ has cost
\begin{equation*}
c_{t_i}((R_i\setminus B_i)\cup S)
= |S| =
\left\lceil\frac{3|R_i\cap A_{t_i}|}{2(n-i+1)}\right\rceil
\le \left\lceil 3 \cdot \mu_{t_i}\right\rceil
= 3 \cdot \mu_{t_i},
\end{equation*}
where the last equality uses the fact that $\mu_{t_i}$ is an integer.
Therefore, it remains to show the existence and computation of a set $S\subseteq B_i$ satisfying $\Phi(B_i\setminus S, n-i) \le \Phi(R_i,n-i+1)$, as this implies $\Phi(R_{i+1},n-i) < 1$ and thus maintains the invariant.

\subsection{The Existence and Computation of Set \texorpdfstring{$S$}{}}

In this subsection, we establish the key technical step by showing how to find a feasible set $S\subseteq B_i$ in round $i$.
As in the previous sections, we first prove the existence of such a set via a random sampling argument (Lemma~\ref{lem:binary-3mms-one-step}), and then show that the sampling process can be derandomized in deterministic polynomial time (Lemma~\ref{lem:binary-3mms-derandomization}).

\begin{lemma}
\label{lem:binary-3mms-one-step}
Consider any $i<n$, and let $t_i\in [k]$ be the type of agent $i$.
Let $B_i = R_i \cap A_{t_i}$.
There exists a set $S\subseteq B_i$ such that
\begin{equation*}
    |S| = \left\lceil\frac{3|B_i|}{2(n-i+1)}\right\rceil
    \quad
    \text{and}
    \quad
    \Phi(B_i\setminus S,n-i)\le \Phi(R_i,n-i+1).
\end{equation*}
\end{lemma}
\begin{proof}
If $B_i=\varnothing$, then choosing $S=\varnothing$ trivially satisfies the lemma.
Hence, we assume that $B_i\ne\varnothing$.
Choose $S$ uniformly at random from all subsets of $B_i$ of size $\left\lceil 3|B_i|/(2(n-i+1))\right\rceil$.
Then the next residual set $R_{i+1}=B_i\setminus S$ has deterministic size
\begin{equation*}
    |R_{i+1}| =
    |B_i|-\left\lceil \frac{3|B_i|}{2(n-i+1)}\right\rceil
    \le
    \left(1-\frac{3}{2(n-i+1)}\right)|B_i|.
\end{equation*}

It therefore remains to show that
$\E[\Phi(R_{i+1},n-i)] \le \Phi(R_i,n-i+1)$,
which immediately implies the existence of the desired set $S$.
Recall that for all $R \subseteq M$ and $r\in [n]$,
\begin{equation*}
    \Phi(R,r)=\sum_{t\in[k]}\sum_{\ell\in[r]}
\binom{|R\cap A_t|}{2\ell\cdot \mu_t+1}
\left(\prod_{j=\ell+1}^{r}\left(1-\frac{3}{2j}\right)\right)^{2\ell\cdot \mu_t+1}.
\end{equation*}

Fix a type $t\in[k]$ and $\ell\in[n-i]$.
We first upper bound
$\E\left[\binom{|R_{i+1}\cap A_t|}{2\ell\cdot \mu_t+1}\right]$.

Notice that, since $R_{i+1}\subseteq B_i$, every subset of $R_{i+1}\cap A_t$ is also a subset of $B_i\cap A_t$.
Therefore, to count subsets of $R_{i+1}\cap A_t$ of a given size, it suffices to analyze the probability that a fixed subset $T\subseteq B_i\cap A_t$ of that size survives the sampling process, namely, that $T\cap S=\varnothing$.

For any fixed subset $T\subseteq B_i\cap A_t$ with $|T|=2\ell\cdot \mu_t+1$, we have
$\Pr[T\cap S=\varnothing]=0$
if $|B_i\setminus S|<|T|$.
Otherwise,
\begin{equation*}
    \Pr[T\cap S=\varnothing]
    =
    \frac{\binom{|B_i|-|T|}{|B_i\setminus S|-|T|}}
         {\binom{|B_i|}{|B_i\setminus S|}}
    \le
    \left(\frac{|B_i\setminus S|}{|B_i|}\right)^{|T|}
    \le
    \left(1-\frac{3}{2(n-i+1)}\right)^{2\ell\cdot \mu_t+1}.
\end{equation*}

Therefore, we have
\begin{align*}
    \E\left[\binom{|R_{i+1}\cap A_t|}{2\ell\cdot \mu_t+1}\right]
    &\le
    \binom{|B_i\cap A_t|}{2\ell\cdot \mu_t+1}
    \left(1-\frac{3}{2(n-i+1)}\right)^{2\ell\cdot \mu_t+1} \\
    &\le
    \binom{|R_i\cap A_t|}{2\ell\cdot \mu_t+1}
    \left(1-\frac{3}{2(n-i+1)}\right)^{2\ell\cdot \mu_t+1}.
\end{align*}

Consequently, we have
\begin{align*}
&\ \E[\Phi(R_{i+1},n-i)] \\
= &\ 
\sum_{t\in[k]}
\sum_{\ell\in[n-i]}
\E\left[\binom{|R_{i+1}\cap A_t|}{2\ell\cdot \mu_t+1}\right]
\left(
\prod_{j=\ell+1}^{n-i}
\left(1-\frac{3}{2j}\right)
\right)^{2\ell\cdot \mu_t+1} \\
\le &\ 
\sum_{t\in[k]}
\sum_{\ell\in[n-i]}
\binom{|R_i\cap A_t|}{2\ell\cdot \mu_t+1}
\left(1-\frac{3}{2(n-i+1)}\right)^{2\ell\cdot \mu_t+1}
\left(
\prod_{j=\ell+1}^{n-i}
\left(1-\frac{3}{2j}\right)
\right)^{2\ell\cdot \mu_t+1} \\
= &\
\sum_{t\in[k]}
\sum_{\ell\in[n-i]}
\binom{|R_i\cap A_t|}{2\ell\cdot \mu_t+1}
\left(
\prod_{j=\ell+1}^{n-i+1}
\left(1-\frac{3}{2j}\right)
\right)^{2\ell\cdot \mu_t+1} \\
\le &\
\sum_{t\in[k]}
\sum_{\ell\in[n-i+1]}
\binom{|R_i\cap A_t|}{2\ell\cdot \mu_t+1}
\left(
\prod_{j=\ell+1}^{n-i+1}
\left(1-\frac{3}{2j}\right)
\right)^{2\ell\cdot \mu_t+1} 
= \Phi(R_i,n-i+1),
\end{align*}
which establishes the existence of the desired set $S$.
\end{proof}

Applying the method of upper bounding the conditional expectations, as in the previous sections, we can derandomize the above sampling procedure and compute the desired set $S$ deterministically in polynomial time.
The proof of the following lemma is deferred to Appendix~\ref{app:omitted-3mms}.

\begin{lemma}
\label{lem:binary-3mms-derandomization}
For any $i < n$, the set $S\subseteq B_i$ described in Lemma~\ref{lem:binary-3mms-one-step} can be computed in deterministic polynomial time.
\end{lemma}

Combining the above analysis yields the following theorem.

\begin{theorem}
\label{thm:known-binary-three-mms}
For $k$ known binary additive cost function types with $k\le n$, there exists a deterministic polynomial-time online algorithm that returns a $3$-MMS allocation.
\end{theorem}

When there are more known types, the initial potential contains more terms, making it harder to keep their sum below one.
Increasing the sampling coefficient above $3/2$ means allocating more costly items in each round.
This reduces the survival weights in the potential, but also raises the current agent's cost bound.
Thus, handling more types in this way may require a larger competitive ratio.
To obtain such a guarantee, the coefficient must be chosen so that the initial potential remains below one and the residual bound is preserved in every round.
\par
%\begin{revision}
%\paragraph{Extension to larger type pools.}
%The condition $k\le n$ can be relaxed within the same framework.For example, increasing the sampling parameter gives a $4$-competitive algorithm for $k\le n^2$.More generally, larger sampling parameters allow larger type pools, at the cost of a larger competitive ratio.For the $4$-competitive extension, the residual bound remains unchanged, and the same conditional-expectation method applies with the modified sample size and potential weights. Appendix~\ref{app:larger-type-pools} gives the details for $k\le n^2$.
%\par
%\end{revision}

\subsection{Lower Bound for Bounded \texorpdfstring{$k$}{}}

We complement our positive results with a lower bound, showing that even when $k=O(1)$, no deterministic online algorithm can guarantee a competitive ratio strictly smaller than $2$.

\begin{theorem}\label{thm:binary-small-k-two-lower}
No deterministic online algorithm can achieve a competitive ratio strictly smaller than $2$ for known binary additive cost instances, even when $k=O(1)$.
\end{theorem}
\begin{proof}
We construct an instance with $k=21$, where $n\ge k$ is divisible by $3$ and $|M|=2n$. Note that $n$ can be chosen arbitrarily large.
Partition the items $M$ into six disjoint groups $G_1,\ldots,G_6$, each of size $n/3$. We define $k = 1+\binom{6}{3}=21$
known binary types. 
The first type, denoted by $t_0$, has cost $1$ for every item; thus its approval set is $A_{t_0}=M$, and $\MMS_{t_0}=\left\lceil \frac{|M|}{n}\right\rceil = 2$.
For every subset $S\subseteq [6]$ of size $3$, define a type $t_S$ whose approval set is
$A_{t_S}=\bigcup_{j\in S}G_j$.
Since $|A_{t_S}|=n$, we have $\MMS_{t_S}=1$ for every such type.

Fix any deterministic online algorithm, and suppose it guarantees $\alpha$-MMS for some $\alpha<2$. The adversary first reveals one agent of type $t_0$. Since $\MMS_{t_0}=2$, this agent can receive at most three items; otherwise, her cost would be at least $4>\alpha\cdot \MMS_{t_0}$. Hence, the allocated bundle $X_1$ intersects at most three of the six groups.

Let $S\subseteq [6]$ be the set of groups intersected by $X_1$. Then $|S|\le 3$, so there exists a subset $T\subseteq [6]\setminus S$ with $|T|=3$. The adversary now makes all remaining $n-1$ agents have type $t_T$. Since no item in $A_{t_T}$ has been allocated to the first agent and $|A_{t_T}|=n$, all $n$ items in $A_{t_T}$ must be allocated among the remaining $n-1$ agents. By the pigeonhole principle, some such agent receives at least two items from $A_{t_T}$, incurring cost at least $2$. This violates the $\alpha$-MMS guarantee because $\MMS_{t_T}=1$ and $\alpha<2$.
\end{proof}

Together, Theorems~\ref{thm:known-binary-three-mms} and~\ref{thm:binary-small-k-two-lower} show that, for known binary additive costs in the small-$k$ regime, the optimal deterministic approximation ratio lies between $2$ and $3$.

We further show in Appendix~\ref{appendix: exactMMS} that exact MMS allocations can be guaranteed when $k=2$, while a $3/2$ lower bound holds for $k=7$. 
We leave open the question of characterizing the values of $k$ for which exact MMS allocations, or more generally constant-factor approximations, are possible.

\section{Conclusion and Open Problems}

We initiated the study of MMS allocations for chores in the online agent-arrival model. In the fully uninformed setting with subadditive cost functions, we presented an algorithm with a competitive ratio of $O(\min\{n, k\log^{1+\epsilon}k, \log m\})$ for any constant $\epsilon>0$, which is near-optimal. For the known-types setting, we showed that polylogarithmic and constant competitive ratios are achievable in the additive and binary additive settings, respectively.

Several gaps in the competitive ratios remain open, particularly in the known-types setting. For example, our results suggest that the optimal competitive ratio for the binary setting grows with $k$. However, the threshold value of $k$ for which constant-competitive or even $1$-competitive algorithms exist is still unclear.
Furthermore, it remains unknown whether the idea of maintaining a universal residual can be extended to the general additive setting to obtain a constant competitive ratio when $k \le n$.
Finally, it would be interesting to investigate whether randomized algorithms can achieve strictly better competitive ratios for the problem, either in expectation or with high probability.
\newpage
\section*{Declaration of AI Use}
The main algorithmic framework was developed entirely by the authors.
The authors acknowledge the use of GPT-5.6 Sol to assist in exploring parts of the proofs.
The manuscript was written by the authors, with ChatGPT and Codex used for language editing and literature searches.
All arguments were independently verified by the authors, who take full responsibility for the manuscript.

\bibliography{ref}
\bibliographystyle{alpha}

\newpage

\appendix

\section{Upper Bounds for Known Number of Types}
\label{sec:known-number-types}

In this appendix, we turn to the setting where the algorithm knows the exact value of $k$ upfront. We show that the competitive ratio can be improved to $\min\{n,k, \lceil \log m \rceil + 1\}$.

As mentioned in Section~\ref{sec:fully-uninformed-upper}, we simply fix $ \ell=\min\{n,k,\lceil \log m\rceil+1\}$ for all agents. Equivalently, since \(L=2(\lceil \log m\rceil+1)\), we have $\ell=\min\{n,k,L/2\}$. In what follows, we prove that all items are allocated.

\begin{algorithm}[!htb]
\caption{$\min\{n, k, \lceil\log m\rceil + 1\}$-MMS Algorithm for Known Number of Types}
\label{alg: k_type}
\KwIn{A set $M$ of items, the number of online agents $n$, and the number of types $k$}

Let $R_1 \gets M$ be the set of remaining items, $\ell\gets \min\{n,k,L/2\}$\;

\For{each online agent $i$, where $i=1,2,\ldots,n$}{
    Compute the MMS partition $\mathbf P^{(i)}=(P^{(i)}_1,\ldots,P^{(i)}_n)$ with respect to $c_i$\;

    Relabel the bundles of $\mathbf P^{(i)}$ so that
    $|P^{(i)}_1\cap R_i|\geq\cdots\geq|P^{(i)}_n\cap R_i|$;

    Let $X_i\gets \bigcup_{j=1}^{\ell}\left(P^{(t_i)}_j\cap R_i\right)$ be union of the $\ell$ bundles with most unallocated items;

    Update the set of remaining items: $R_{i+1} \gets R_i\setminus X_i$;
}
\KwOut{Allocation $\mathbf X=(X_1,\ldots,X_n)$}
\end{algorithm}

\begin{lemma} \label{lemma: knownk_kmms}
    With known $k$, Algorithm~\ref{alg: k_type} outputs a $\min\{n,k,\lceil \log m \rceil + 1\}$-MMS allocation.
\end{lemma}
\begin{proof}
    We first prove the cost guarantee. Similar to the analysis in Section~\ref{sec:fully-uninformed-upper}, we consider an arbitrary agent $i$ and have
\begin{align*}
    c_i(X_i)
    \leq \sum_{j=1}^{\ell} c_i(P^{(t_i)}_j\cap R_i)
    \leq \sum_{j=1}^{\ell} c_i(P^{(t_i)}_j)
    \leq  \ell\cdot \MMS_i.
\end{align*}
Thus every agent receives a cost at most $\ell \cdot \MMS_i.$
    It remains to show that no item is left unallocated. We consider the following cases according to the value of $\ell$.
    \begin{itemize}
        \item 
        If $\ell = n$, then the first agent will take all items, which means all items are allocated.
        \item 
        If $\ell = L/2$, we follow the analysis in Section~\ref{sec:fully-uninformed-upper} and have
        \begin{equation*}
            |R_{i+1}| = |R_i \setminus X_i| \leq \left(1-\frac{L}{2n}\right)\cdot | R_i|,
        \end{equation*}
        which implies that
        \begin{equation*}
            |R_{n+1}| \leq \left(1-\frac{L}{2n} \right)^{n} \cdot |R_1| \leq m \cdot \exp\left(-\frac{L}{2}\right) < 1.
        \end{equation*}
        Thus, all items are allocated.
        \item 
        If $\ell = k$, we prove that all items are allocated via a pigeonhole argument. Since there are $n$ agents distributed among $k$ distinct types, by the pigeonhole principle, at least one type $t$ appears at least $\left\lceil n/k\right\rceil$ times. For this type $t$, Algorithm~\ref{alg: k_type} allocates $\lceil n/k\rceil \cdot \ell = \lceil n/k\rceil \cdot k \geq n$ distinct bundles from the MMS partition with respect to $c_t$. Thus, no item is left unallocated.
    \end{itemize}

    Hence, Algorithm~\ref{alg: k_type} outputs a $\min\{n,k,\lceil \log m\rceil + 1\}$-MMS allocation.
\end{proof}

\section{Lower Bounds for Known Subadditive Types}
\label{app:known-subadditive-lower}

In this appendix, we provide an $\Omega\left( \min \{n, k, \log m/\log\log m\} \right)$ offline lower bound for known subadditive type-cost functions. The $n$- and $m$-dependent terms follow from the offline lower bound of Li et al.~\cite{conf/nips/0037WZ23}. We therefore focus on the $k$-dependent term for $k\leq n$.

\begin{lemma}
\label{lem:known-types-subadditive-k-lower}
For every $k\leq n$, there is an instance of chore allocation with $n$ agents and $k$ subadditive cost types such that every complete allocation has MMS approximation ratio $\Omega(k)$.
%In the online chores allocation with agent arrivals, if the agents belong to $k$ ($k<n$) types and the cost functions are known submodular functions, no algorithm can achieve a competitive ratio of $o(k)$.
\end{lemma}

\begin{proof}
We first describe our instance. Let $q=\lceil n/k\rceil$ and $B=q\cdot k$. There are $|M|=B^k$ items in the instance, and each item is represented by a vector $\mathbf{x}=(x_1,\ldots,x_k)$, where $x_i \in [B]$. For each $t\in[k]$, we define the cost function as
\begin{align*}
    c_t(S) = \left| \left\{ \ell\in[B]: \exists x\in S \text{ such that } x_t=\ell \right\} \right|.
\end{align*}
Here, $c_t(S)$ counts the number of distinct values that appear in the $t$-th coordinate among the items in $S$. It is easy to verify that this cost function is subadditive. We use $r_t$ to denote the number of agents of type $t$. Since $q\cdot k\ge n$, there exist positive integers $r_1,\dots,r_k$ with $\sum_{t=1}^k r_t=n$ and $r_t\le q$ for all $t$. We fix such integers $r_1,\ldots,r_k$.

We show that $\MMS_t$ is upper bounded by $2$ for every $t \in [k]$ in this instance. For any type $t$, the items in $M$ can be partitioned into $B$ subsets according to their $t$-th coordinate; all items in a subset share the same value of this coordinate. Since $B = q\cdot k \leq n+k \le 2n$, these subsets can be grouped into $n$ bundles, each containing at most two subsets. For any type $t$, every such bundle has a cost of at most $2$, therefore we have
\begin{equation*}
    \MMS_t\leq 2 \quad \text{for every }t\in[k].
\end{equation*}

In the following, we prove that an $o(k)$-MMS allocation is impossible by showing that, for any allocation, there is at least one agent who receives a bundle with cost at least $k$. Suppose, toward a contradiction, that every agent of type $t$ receives a bundle of cost at most $k-1$ (the cost must be an integer). For any type $t$, the bundles assigned to all type-$t$ agents contain items with at most $r_t\cdot (k-1)$ distinct values in the $t$-th coordinate. Since $r_t\cdot (k-1) \leq q\cdot (k-1) < B$, for each coordinate dimension $t \in [k]$, there is at least one available coordinate $\ell_t \in [B]$ that is completely untouched by the bundles assigned to type-$t$ agents. By combining these missing coordinates, it follows that the item $x^* = (\ell_1, \dots, \ell_k)$ must remain unallocated, which leads to a contradiction. Therefore, some agent incurs a cost of at least $k$, implying that the MMS approximation ratio of the allocation is at least $k/2$.
\end{proof}

Applying the offline lower bound by Lemma~\ref{lem:known-types-subadditive-k-lower} and the results of Li et al.~\cite{conf/nips/0037WZ23}, we get the following result.
\begin{corollary}
\label{cor:known-types-subadditive-lower}
In the online chores allocation with agent arrivals, if the agents belong to $k$ types and the cost functions are known subadditive functions, no algorithm can achieve a competitive ratio of
\begin{equation*}
    o\!\left(\min\left\{n,\frac{\log m}{\log\log m},k\right \} \right).
\end{equation*}
\end{corollary}

\section{Derandomization for \texorpdfstring{$O(\log k)$}{}-MMS Universal Partition}
\label{app:additive-derandomization}

In this appendix, we present the derandomization for the $O\!\left(\log k\right)$-MMS universal partition.

\begin{lemma}
\label{lem:additive-extraction-derandomization}
The randomized extraction step in Lemma~\ref{lem:additive-logk-extraction} can be implemented deterministically in polynomial time.
\end{lemma}
\begin{proof}

Suppose Invariant~\ref{invariant:logk_universal} holds at the beginning of round $i$. We now show how to compute a set $S$ deterministically in polynomial time such that the invariant holds at the beginning of round $i+1$.

Fix an arbitrary ordering of the remaining items. We determine whether each item is assigned to $S$ sequentially while maintaining an upper bound of $1$ on the conditional expectation of the potential. In particular, let $\E[\Phi(S) \mid (F, U)]$ denote the conditional expectation of the potential given that the items in $F$ have been fixed to be included in $S$ and the items in $U$ have been fixed to be excluded from $S$. The expectation is taken over the random assignment of the remaining undetermined items in $R_i \setminus (F\cup U)$. For every item $e \in R_i \setminus(F\cup U) $, it is assigned to $S$ with probability $p=6/(n-i+1)$ and not assigned to $S$ with probability $1-p$. As shown in Section~\ref{sec:known-types-additive-logk}, we have
\begin{equation*}
    \E[\Phi(S)\mid (\varnothing, \varnothing)] < 1.
\end{equation*}
Moreover, when $F\cup U = R_i$, we have
\begin{equation*}
    \E[\Phi(S) \mid (F, U)] = \Phi(F).
\end{equation*}

In the following, we maintain an invariant that $\E[\Phi(S)\mid (F, U)] < 1$ throughout the assignment process. Since the initial state satisfies the invariant, its preservation across round $i$ guarantees the resulting set $S$ is what we need.

Suppose that the items in $F$ have already been fixed to be included in $S$ and the items in $U$ have been fixed to be excluded from $S$, and let $e$ be the next undecided item. We define $\Psi_F$ as the conditional expectation of $\Phi(S)$ given $F, U$ and $e$ is assigned to $S$:
\begin{equation*}
    \Psi_F := \E[\Phi(S) \mid (F\cup \{e\}, U)].
\end{equation*}
Similarly, we define $\Psi_U$ as the conditional expectation of $\Phi(S)$ given $F, U$ and $e$ is not assigned to $S$:
\begin{equation*}
    \Psi_U := \E[\Phi(S) \mid (F, U\cup \{e\})].
\end{equation*}
Since $e$ is assigned to $S$ with probability $p$, we have
\begin{equation*}
    \E[\Phi(S)\mid (F, U)] = p\cdot \Psi_F + (1-p)\cdot \Psi_U.
\end{equation*}
Either $\Psi_F$ or $\Psi_U$ is at most $\E[\Phi(S)\mid (F, U)]$, and therefore less than $1$. If $\Psi_F \leq \E[\Phi(S)\mid (F, U)]$, we assign $e$ to $S$; otherwise, we do not assign $e$ to $S$. Thus, the invariant $\E[\Phi(S)\mid (F, U)] < 1$ is preserved, and we continue the assignment until $F \cup U = R_i$.

It remains to show that, given any two sets $F$ and $U$, the conditional expectation $\E[\Phi(S)\mid (F,U)]$ can be computed in polynomial time. From the definition of the potential function, it suffices to compute $\E[2^{c_t(S)}]$ for every type $t\in [k]$ and $\E[5^{-c_t(S)}]$ for every critical type $t \in A$. For $\E[2^{c_t(S)}]$, we have
\begin{equation*}
   \E[2^{c_t(S)}]= 2^{c_t(F)} \cdot \prod_{e \in R_i\setminus(F\cup U)} (1-p + p\cdot 2^{c_t(e)}).
\end{equation*}
For $\E[5^{-c_t(S)}]$, we have
\begin{equation*}
    \E[5^{-c_t(S)}] = 5^{-c_t(F)} \cdot \prod_{e \in R_i\setminus(F\cup U)} (1-p + p\cdot 5^{-c_t(e)}).
\end{equation*}
Both expressions are clearly computable in polynomial time, which completes the proof.
\end{proof}

\section{Omitted Proofs from Section~\ref{sec:known-types-binary}}
\label{app:omitted-3mms}

In this appendix, we prove the technical statements used in Section~\ref{sec:known-binary-three-mms}.

\begin{claim}
For all $1\le \ell\le r$,
$
\prod_{i=\ell+1}^{r}\left(1-\frac{3}{2i}\right)
\le
\left(\frac{\ell}{r}\right)^{3/2}.
$
\end{claim}
\begin{proof}
It is enough to show that for every $i\ge 2$,
$
1-\frac{3}{2i}\le \left(\frac{i-1}{i}\right)^{3/2}.
$
Let $x=1/i$.
It suffices to show that
$$
(1-x)^{3/2}-1+\frac{3x}{2}\ge 0
$$
for $x\in[0,1/2]$.
The derivative of the left-hand side is
$
\frac{3}{2} \cdot \left(1-\sqrt{1-x}\right),
$
which is nonnegative on this interval.
Since the left-hand side is zero at $x=0$, the inequality follows.
Multiplying the inequalities for $i=\ell+1,\ldots,r$ gives
\begin{equation*}
\prod_{i=\ell+1}^{r}\left(1-\frac{3}{2i}\right)
\le
\prod_{i=\ell+1}^{r}\left(\frac{i-1}{i}\right)^{3/2}
=
\left(\frac{\ell}{r}\right)^{3/2}. \qedhere
\end{equation*}
\end{proof}

\begin{claim}
For every integer $N\ge 1$,
$$
\sum_{i=3}^{N}
\binom{N}{i}
\left(\frac{i}{2N}\right)^{3i/2}
<
\frac{1}{N}.
$$
\end{claim}

\begin{proof}
The sum is $0$ for $N\le 2$.
Assume $N\ge 3$ and write
$$
a_i
=
\binom{N}{i}
\left(\frac{i}{2N}\right)^{3i/2}.
$$
For $3\le i<N$,
\begin{align*}
\frac{a_{i+1}}{a_i}&= \frac{N-i}{i+1} \cdot \left(\frac{i+1}{2N}\right)^{3/2} \left(1+\frac{1}{i}\right)^{3i/2} \\
&=(N-i) \cdot \left(\frac{1}{2N}\right)^{3/2} \cdot (i+1)^{1/2} \cdot \left(1+\frac{1}{i}\right)^{3i/2}.
\end{align*}
Using $(1+1/i)^i<e$ and $i+1\le 4i/3$ for $i\ge 3$, we get
\begin{align*}
\frac{a_{i+1}}{a_i}
&< (N-i) \cdot  \left(\frac{1}{2N}\right)^{3/2} \cdot \frac{2\sqrt{i}}{\sqrt{3}} \cdot e^{3/2} \\
&= \frac{e^{3/2}}{2^{3/2}}\cdot\frac{2}{\sqrt{3}}\cdot
\left(1-\frac{i}{N}\right)\sqrt{\frac{i}{N}}.    
\end{align*}
The function $(1-x)\sqrt{x}$ has maximum $2/(3\sqrt{3})$ on $[0,1]$.
Hence
$$
\frac{a_{i+1}}{a_i}
<
\frac{e^{3/2}}{2^{3/2}}\cdot\frac{2}{\sqrt{3}}\cdot\frac{2}{3\sqrt{3}}
<
\frac{3}{4}.
$$
Thus, the tail is dominated by a geometric series with ratio $3/4$.
For $N\ge 18$,
$$
a_3
\le
\frac{N^3}{6}\cdot
\left(\frac{3}{2N}\right)^{9/2}
<
\frac{1}{4N}.
$$
It follows that
$$
\sum_{i=3}^{N}a_i
\le
\frac{a_3}{1-3/4}
<
\frac{1}{N}.
$$
It remains to check the finite range $3\le N\le 17$.
For this range, it is enough to verify that $N\cdot\sum_{i=3}^{N}a_i<1$.
This is a finite check over the fifteen values $N=3,\ldots,17$, and direct evaluation verifies that $N\cdot \sum_{i=3}^{N}a_i<1$ in each case.
Equivalently, $\sum_{i=3}^{N}a_i<1/N$ for every $3\le N\le 17$.
This proves the claim.
\end{proof}

\begin{lemma} 
For any $i < n$, the set $S\subseteq B_i$ described in Lemma~\ref{lem:binary-3mms-one-step} can be computed deterministically in polynomial time.
\end{lemma}

\begin{proof}
Before we present the derandomization, we recall some notions in Lemma~\ref{lem:binary-3mms-one-step}. Fix a round $i < n$. Let $t_i$ be the type of the arriving agent, let $R_i$ be the set of remaining items, and let $B_i = A_{t_i} \cap R_i$ be the set of remaining items that have cost $1$ for the arriving agent. Lemma~\ref{lem:binary-3mms-one-step} states that randomly selecting $s=\left\lceil 3\cdot|B_i|/(2\cdot (n-i+1))\right\rceil$ items from $B_i$ yields a desired set $S$ with positive probability. 

We now describe an equivalent sequential formulation of the same random process. 
Fix an arbitrary ordering of the items in $B_i$ and process them one by one. 
We maintain two disjoint sets $F$ (items already committed to $S$) and $U$ (items already committed to be excluded from $S$). 
Initially $F=U=\varnothing$. 
When the next item is considered, let $R'=B_i\setminus(F\cup U)$ be the set of remaining undecided items and let $q=\max\{s-|F|,0\}$ be the number of items still required to reach the target size $s$. 
The item is then placed into $S$ (i.e., added to $F$) with probability $q/|R'|$ and excluded from $S$ (i.e., added to $U$) with probability $1-q/|R'|$.

First, we show that this sequential procedure is well-defined. That is, the probability $q/|R'|$ is always at most $1$. When the sequential procedure starts, since $i<n$, we have $s \leq |B_i|$. Thus, the probability is well-defined for the first item. Then we suppose that for some item (not the first item), the probability $q/|R'|$ exceeds $1$ for the first time. Since after deciding an item, either $q$ and $|R'|$ both decrease by $1$ or only $|R'|$ decreases by $1$, we must have $q = |R'| + 1$. This implies that in the previous round, all items in $R'$ have been assigned to $S$ with probability $1$. So the sequential procedure ends in the previous round, which leads to a contradiction. Thus, this random procedure is well-defined.

We next show that the sequential procedure is equivalent to selecting a uniformly random $s$-subset of $B_i$ by coupling it to a random permutation. Consider a uniformly random permutation of items in $B_i$, and let $S$ be the set consisting of the first $s$ items in this permutation. 
Clearly, $S$ is a uniformly random $s$-subset of $B_i$. 
Now reveal the items' positions one by one in this order. 
Suppose we have revealed $j$ items, and items in $F$ have already been revealed and committed to $S$, leaving $q$ additional items still required. 
Among the $|B_i|-j+1$ unrevealed items, exactly $q$ occupy positions in the first $s$ of the permutation. 
Hence, the next revealed item belongs to $S$ with conditional probability $q/(|B_i|-j+1)$, which is precisely the sequential rule described above. 
Therefore, the sequential procedure generates the same distribution as selecting a uniformly random $s$-subset.

In what follows, we show how to derandomize this sequential procedure to get the set $S$ deterministically. Given two disjoint sets $F$ and $U$, we define the conditional expectation of the potential in Lemma~\ref{lem:binary-3mms-one-step} as
\begin{align*}
  \E[\Phi(B_i\setminus S, n-i)\mid (F, U)],
\end{align*}
where the expectation is taken over the random set $S$, which is a superset of $F$ and may contain more items from $R'$.
Initially, $F=U=\varnothing$, by Lemma~\ref{lem:binary-3mms-one-step}, we have
\begin{align*}
   \E[\Phi(B_i\setminus S, n-i)\mid (\varnothing, \varnothing)]  = \E[\Phi(B_i\setminus S, n-i)] \leq \Phi(R_i, n-i+1).
\end{align*}

When $|F| = s$, there is no randomness on $S$, and we have
\begin{align*}
    \E[\Phi(B_i\setminus S, n-i)\mid (F, U)] = \Phi(B_i\setminus F, n-i).
\end{align*}

In the following, we maintain the invariant that
\begin{equation*}
    \E[\Phi(B_i\setminus S, n-i) \mid (F, U)] \leq \Phi(R_i, n-i+1)
\end{equation*}
throughout the sequential procedure. Since the initial state satisfies the invariant, its preservation across round $i$ guarantees the resulting set $S$ is what we desire.

Suppose that the items in $F$ have already been fixed to be included in $S$ and the items in $U$ have been fixed to be excluded from $S$, and let $e$ be the next undecided item. We define $\Psi_F$
as the conditional expectation of $\Phi(B_i\setminus S, n-i)$ given $F, U$ and $e$ is assigned to $S$:
\begin{align*}
    \Psi_F = \E[\Phi(B_i\setminus S, n-i)\mid (F\cup\{e\}, U)].
\end{align*}

Similarly, we define $\Psi_U$ as the corresponding expectation when $e$ is not assigned to $S$:
\begin{equation*}
    \Psi_U = \E[\Phi(B_i\setminus S, n-i)\mid (F, U\cup\{e\})].
\end{equation*}

Since $e$ is assigned to $S$ with probability $q/|R'|$ in $\E[\Phi(B_i\setminus S, n-i)\mid (F, U)]$, we have
\begin{align*}
    \E[\Phi(B_i\setminus S, n-i)\mid (F, U)] &= \frac{q}{|R'|}\cdot \Psi_F + \left(1-\frac{q}{|R'|}\right) \cdot \Psi_U.
\end{align*}
Therefore, at least one of $\Psi_F$ and $\Psi_U$ is at most $\E[\Phi(B_i\setminus S, n-i)\mid (F, U)]$, and hence at most $\Phi(R_i, n-i+1)$. If $\Psi_F \leq \E[\Phi(B_i\setminus S, n-i) \mid (F, U)]$, we assign $e$ to $S$; otherwise, we do not assign $e$ to $S$. Thus, the invariant $\E[\Phi(B_i\setminus S, n-i) \mid (F, U)] \leq \Phi(R_i, n-i+1)$ is preserved, and we continue the assignment until $|F| = s$.

It remains to show that $\E[\Phi(B_i\setminus S, n-i) \mid (F, U)]$ can be computed in polynomial time for any possible sets $F$ and $U$.
Fix sets $F, U$, a type $t\in[k]$, and a level $\ell\in[n-i]$. Let $h=2\ell\cdot \mu_t+1.$ By the definition of $\Phi$, for this fixed pair $(t,\ell)$, the only random part of the corresponding term is $\binom{|(B_i \setminus S)\cap A_t|}{h}.$
Therefore, it is enough to compute $\E\left[\binom{|(B_i \setminus S)\cap A_t|}{h}\right]$, where the expectation is over $S$, conditioned on $F\subseteq S$ and $S\cap U = \varnothing$.

Since the items in $U$ are excluded from $S$, every size-$h$ subset of $(B_i \setminus S)\cap A_t$ can be classified by the number of its items that come from $U\cap A_t$. With a slight abuse of notation, we suppose there are exactly $j$ items of a size-$h$ subset that come from $U \cap A_t$. Thus, for this size-$h$ subset, the remaining $h-j$ items must come from $R'\cap A_t$. There are 
\begin{align*}
    \binom{|U\cap A_t|}{j}
    \binom{|R'\cap A_t|}{h-j}
\end{align*}
ways to choose these items.

The $j$ items chosen from $U\cap A_t$ always belong to $B_i \setminus S$.
The $h-j$ items chosen from $R'\cap A_t$ belong to $B_i \setminus S$ only if none of them is selected into $S$.
Since we can regard the remaining random selection as uniformly picking a size-$q$ subset of $R'$, this probability is $\binom{|R'|-(h-j)}{q} / \binom{|R'|}{q}$.

Therefore, we have
\begin{align*}
\mathbb E\left[\binom{|(B_i \setminus S)\cap A_t|}{h}\right]
= \sum_{j=0}^{h} \binom{|U\cap A_t|}{j}
\binom{|R'\cap A_t|}{h-j}
\cdot
\frac{\binom{|R'|-(h-j)}{q}}{\binom{|R'|}{q}}.
\end{align*}

For this fixed pair $(t,\ell)$, the other factor in the term of $\Phi$ is fixed, so the whole term can be computed.
Thus, after precomputing all binomial coefficients $\binom{x}{y}$ for $0\le y\le x\le m$, each value of $\E[\Phi(B_i\setminus S, n-i) \mid (F, U)]$ can be computed in polynomial time.

Finally, the fixing process makes at most $|B_i|\leq m$ decisions, and each decision requires computing at most two values of $\Phi$. Therefore, the desired set $S$ can be computed deterministically in polynomial time.
\end{proof}

\section{The Existence of Exact MMS for Binary Costs}\label{appendix: exactMMS}

In this appendix, we study exact MMS allocations when the agents belong to a constant number of known binary additive types.
We first show that exact MMS allocations can be computed when $k = 2$.
We then show that when $k = 7$, no deterministic online algorithm can guarantee an approximation ratio strictly smaller than $3/2$.

The case $k=1$ is immediate.
We therefore begin with the case $k=2$.
\begin{lemma} \label{lemma: exactMMS_k = 2}
    For $k=2$, there exists a deterministic online algorithm that computes an exact MMS allocation for binary additive costs.
\end{lemma}
\begin{proof}
Let the two known types have approval sets $A_1, A_2\subseteq M$, where an item has cost $1$ for type $t$ if and only if it belongs to $A_t$.
For each $t\in\{1,2\}$, let $\mu_t=\MMS_t=\left\lceil |A_t|/{n}\right\rceil$.
Recall that $R_i$ denotes the set of unallocated items at the beginning of round $i$, let $r=n-i+1$ be the number of remaining agents, including agent $i$.

We maintain the invariant that, at the beginning of every round $i$,
$$
|R_i\cap A_t|\le r\cdot \mu_t \qquad \text{for every } t\in\{1,2\}.
$$
The invariant holds initially because $R_1=M$ and $|A_t|\le n\cdot \mu_t$ for each type $t$.

Fix a round $i<n$.
By symmetry, suppose that the arriving agent has type $1$.
The algorithm first allocates all items in $R_i\setminus A_1$ to the current agent.
These items have zero cost for her.
Among the items that are still unallocated, define
\begin{align*}
    S_1&=R_i\cap(A_1\setminus A_2), S_{12}=R_i\cap A_1\cap A_2.
\end{align*}
Thus, $S_1\cup S_{12}$ is exactly the set of remaining items that are costly for type $1$.
We further define $\Delta_1$ to be the minimum number of additional type-$1$ costly items that must be removed in this round in order to maintain the invariant for type $1$.
Equivalently, 
$$ \Delta_1  =  \max\left\{ 0,\,|S_1|+|S_{12}|-(r-1)\cdot \mu_1 \right\}.$$
Since the invariant holds at the beginning of the round, we have
\begin{align*}
    |S_1|+|S_{12}|=|R_i\cap A_1|\le r\cdot \mu_1,
\end{align*}
and hence $\Delta_1\le \mu_1.$

The algorithm first chooses a set
    $T\subseteq S_{12}$ with
    $|T|=\min\{|S_{12}|,\mu_1\}.$
The items in $T$ are costly for both types, so allocating them
helps maintain the invariant for both type $1$ and type $2$.
If $|T|\ge \Delta_1$, then the items in $T$ already remove
enough type-$1$ costly items.
Otherwise, $|T|<\Delta_1$.
Since $\Delta_1\le\mu_1$, this implies that
$|T|=|S_{12}|$.
The algorithm then additionally chooses
$\Delta_1-|T|$ items from $S_1$.
This is possible because
   $ \Delta_1
    \le
    |S_1|+|S_{12}|,$
and hence
   $ \Delta_1-|T|
    =
    \Delta_1-|S_{12}|
    \le
    |S_1|.$

Therefore, the number of costly items allocated to the current
agent is
$    \max\{|T|,\Delta_1\}
    \le
    \mu_1.$
Thus, her cost is at most her MMS.
We next verify the invariant for type $1$.
The algorithm removes exactly
   $ \max\{|T|,\Delta_1\}$
items from $S_1\cup S_{12}$.
Therefore,
\begin{align*}
    |R_{i+1}\cap A_1|
    &= |S_1|+|S_{12}|-\max\{|T|,\Delta_1\}\\
    &\le
    |S_1|+|S_{12}|-\Delta_1 \le (r-1)\mu_1,
\end{align*}
where the last inequality follows from the definition of $\Delta_1$.

For type $2$, every item in $R_i\cap(A_2\setminus A_1)$ was allocated to the current agent at zero cost.
The only type $2$ costly items that may still be unallocated are therefore the items in $S_{12}\setminus T$.
Thus, $|R_{i+1}\cap A_2|=|S_{12}|-|T|.$ 

To show that the invariant for type $2$ is also maintained, we define 
\begin{equation*}
    \Delta_2 = \max \{ 0,\, |S_{12}|-(r-1)\mu_2\}.
\end{equation*}

By definition, $\Delta_2\le |S_{12}|$. We next show that $\Delta_2\le \mu_1$. 
The invariant at the beginning of the round implies that $|S_{12}| \le |R_i\cap A_1 \cap A_2| \le r\cdot \min\{\mu_1, \mu_2\}$.
Therefore, we have
\begin{equation*}
\Delta_2 \le \max\{0,\, r\cdot \min\{\mu_1,\mu_2 \} - (r-1)\cdot \mu_2\}\le \mu_1,
\end{equation*}
which implies
\begin{align*}
    \Delta_2
    \le
    \min\{|S_{12}|,\mu_1\}
    =
    |T|.
\end{align*}
It follows that
\begin{align*}
    |R_{i+1}\cap A_2|
    &=
    |S_{12}|-|T| \le
    |S_{12}|-\Delta_2 \le
    (r-1) \cdot \mu_2.
\end{align*}
Therefore, the invariant is maintained for both types.

By induction, the invariant holds at the beginning of every round, and every nonfinal agent receives a bundle whose cost is at most her MMS.
In the final round, we have $r=1$.
The invariant gives $|R_i\cap A_t|\le \mu_t$, where $t$ is the type of the last agent.
Thus, allocating all remaining items to the last agent incurs a cost of at most $\mu_t$.
Therefore, the algorithm computes an exact MMS allocation for $k=2$.
\end{proof}

We next show that exact MMS allocations cannot always be guaranteed when $k=7$.

\begin{lemma} \label{lemma: MMSnotexist}
No deterministic algorithm can achieve an approximation ratio strictly better than $3/2$, even when $k = 7$.
\end{lemma}
\begin{proof}
Consider an instance with $n = 8$ agents, $k = 7$ agent types, and $m = 16$ items. Let the set of items be partitioned into four disjoint groups, $S_1, S_2, S_3$, and $S_4$, each containing exactly $4$ items.
Next, we define the $7$ agent types.

Let $N_j = \{t \in [k] \mid c_t(g) = 1 \text{ for all } g \in S_j\}$ denote the set of agent types that incur a cost of $1$ for items in $S_j$.
Let $N_1 = \{1, 2, 3, 4\}$, $N_2 = \{1, 2, 5, 6\}$, $N_3 = \{1, 3, 5, 7\}$, and $N_4 = \{1, 4, 6, 7\}$.
In other words, type $t$ agents have cost $1$ on items in $S_j$ if and only if $t\in N_j$.
It can be verified that $\MMS_1 = 2$, and $\MMS_i = 1$ for all $i \in \{2,3,\ldots,7\}$.
We show that no deterministic algorithm can achieve an approximation ratio strictly below $3/2$ for this instance. 

Suppose the adversary first reveals an agent of type $1$. 
To maintain a ratio strictly below $3/2$, the algorithm can allocate at most two items to this agent. 
Let $S_i$ and $S_j$ be the sets from which the two items are chosen, where it is possible that $i = j$.
By the design of the instance, there exists an agent type $t \in [k]$ such that $t \notin N_i\cup N_j$.  
Let the remaining $7$ agents be of type $t$.
Since there are $8$ items with cost $1$ for these agents, by the pigeonhole principle, at least one agent must receive at least two items, thereby incurring a cost of $2$. 
Consequently, this example rules out any deterministic guarantee strictly below $3/2$. 
\end{proof}

\end{document}